%% file: ULTWCFirstRevision_v2.tex
\documentclass[12pt, draftclsnofoot, onecolumn]{IEEEtran}
\pdfoutput=1
\usepackage[font=small]{caption}

\usepackage{amsmath}
\usepackage{mathtools}
\usepackage{amsmath}
\usepackage{amssymb}
\usepackage{cite}
\usepackage{array}
\usepackage{amsthm}
\usepackage{amsmath}
\usepackage{tabulary}
\usepackage{multirow}
\usepackage{subfig}
\usepackage{enumerate}
\IEEEoverridecommandlockouts
\usepackage{pifont}
\usepackage{epsfig}
\usepackage{graphicx}
\usepackage{url}
\usepackage{amssymb}
\usepackage{paralist} %
\usepackage{bbm}

\graphicspath{ {./figs/} }

\include{notation}

\begin{document}
\title{ Joint Rate and SINR Coverage Analysis for Decoupled Uplink-Downlink Biased Cell Associations in HetNets}

\author{\IEEEauthorblockN{Sarabjot Singh, Xinchen Zhang, and Jeffrey G. Andrews}
\thanks{ J. G. Andrews ({\tt jandrews@ece.utexas.edu}) is  with Wireless Networking and Communications Group (WNCG), The University of Texas at Austin, Austin, TX, USA.
S. Singh ({\tt sarabjot@utexas.edu}) and X. Zhang ({\tt xzhang7@alumni.nd.edu}) were also with WNCG. They are now with Nokia Tech., Berkeley, CA, USA and  Qualcomm Inc., San Diego, CA, USA respectively. 
}}
\maketitle

\begin{abstract}
Load balancing by proactively offloading users onto small and otherwise lightly-loaded cells  is critical for tapping the potential of  dense   heterogeneous cellular networks (HCNs). Offloading has mostly been studied for the downlink, where it is generally assumed that a user offloaded to a small cell will communicate with it on the uplink as well. The impact of coupled downlink-uplink offloading is not well understood.  Uplink power control and spatial interference  correlation further complicate  the mathematical analysis as compared to the downlink.  We propose an accurate and  tractable model to characterize the uplink $\SINR$ and rate distribution in a multi-tier  HCN  as a function of the association rules and power control parameters. Joint uplink-downlink rate coverage is also characterized.  Using the developed analysis, it is shown that the optimal degree of channel inversion (for uplink power control) increases with load imbalance in the network.  In sharp contrast to the downlink, minimum path loss  association  is shown to be optimal for uplink rate. Moreover, with minimum path loss  association and full channel inversion, uplink $\SIR$ is shown to be invariant of infrastructure density. It is further shown that  a \emph{decoupled association}---employing differing association strategies for uplink and downlink---leads to significant improvement in joint uplink-downlink rate coverage over the standard  coupled association in HCNs.
\end{abstract}

\section{Introduction}
Supplementing existing cellular networks with low power access points (APs), generically referred to as  small cells, leads to   wireless networks that are  highly heterogeneous in AP max transmit powers  and deployment density \cite{ghosh2012heterogeneous,qcom_hetnet_wmag}. Although the mathematical modeling and performance analysis -- particularly for downlink -- for  HCNs  has received  significant attention in recent years (see \cite{ElSawy13tut} for a survey),    attempts to model and analyze the uplink have been limited.   
In  popular uplink intensive services like cloud storage    and video chat, uplink performance is as important (if not  more) as that of the downlink. Moreover, in  services like video chat, the traffic is symmetric and thus  what really matters    is the ability to achieve   the required   QoS \emph{both} in uplink and downlink. The  insights for downlink design cannot  be directly  extrapolated  to the uplink setting in HCNs, as the latter is fundamentally different  due to \begin{inparaenum}[(i)] \item the homogeneity of  transmitters or user equipments (UEs), \item the   use of uplink transmission power control to the desired AP, and \item the  correlation of the  interference power from a UE   with its path loss to its own serving AP. \end{inparaenum} 

\subsection{Background and related work}
\textbf{Load balancing and power control.}
Due to the large AP transmission power disparity across different tiers in HCNs, the nominal UE load per AP (under downlink maximum power  association) is highly  imbalanced, with macrocells being significantly more congested than small cells.   It is  now   well established (both empirically and theoretically) that biasing UEs towards small cells leads to  significant improvement in downlink throughput (see \cite{ghosh2012heterogeneous,AndLoadCommag13,qcom_hetnet_wmag} and references therein). 
In conventional homogeneous macrocellular  networks, \textit{coupled} associations are used, wherein the UE is paired with the same AP for both uplink and downlink transmission.  Traditionally, this association has been based on the maximum downlink received power as measured at the UE, which also led to a max-uplink power association with the same AP, since the downlink and uplink channels are nearly reciprocal in terms of shadowing and path loss and all APs and UEs had essentially the same transmit powers, respectively.  However, this is clearly not the case in HCNs with load balancing. Biasing UEs towards small cells with a coupled association not only improves the downlink rate (despite a lower SINR) due to the load balancing aspect, but it simultaneously improves the uplink signal--to--noise--ratio ($\SNR$).  This is because the offloaded UEs  now on average transmit to APs, which are closer, since they are more likely to transmit to a nearby small cell whose downlink power was not large enough to associate with in the absence of biasing.   It is dubious, though, whether the bias designed to encourage downlink offloading would also be optimal for the uplink. 

Since transmit power is a critical resource at a UE, power control is  employed to  conserve energy  and also to reduce interference.
3GPP LTE networks support the use of  fractional power control (FPC), which   partially compensates for path loss \cite{XiaUL}.  In FPC, a UE   with path loss  $\pl$  to its serving AP  transmits with power $  \pl^{\epsilon}$, where $0\leq \epsilon\leq 1$  is the power control fraction (PCF).  Thus, with $\epsilon =0$, each UE transmits with constant  power,  and  with $\epsilon=1$, the path loss is fully compensated corresponding to channel inversion.   
 From a network point of view,  $\pcf$  can be interpreted as a fairness parameter, where a higher PCF $\pcf$  helps the cell edge users meet their $\SINR$ target but generates higher interference  \cite{mulUL,MuhULPC,CasUL,SimULPC,NovUL,Cou2011}.  Since the association strategy influences the statistics of  path loss  in    HCNs, the aggressiveness of power control should be  correlated with the association strategy. Therefore, it is important to develop an analytical model to capture the  interplay between  load balancing and power control on the uplink performance.  This is one of the goals of this paper.

\textbf{Uplink analysis.} The use of spatial point processes, particularly  the homogeneous Poisson point process (PPP), for modeling HCNs and derivation of the corresponding downlink coverage and rate under various association and interference coordination strategies has been extensively explored as of late (see \cite{ElSawy13tut} and references therein). The homogeneous  PPP assumption for AP location  not only greatly simplifies the downlink interference characterization, but also comes with  empirical and theoretical support  \cite{andganbac11,BlaKarKee12,ADG_Tcom,SinDhiAnd13}. However,  analysis of the uplink  in such a setting is highly non-trivial, as the uplink interference does not originate from   Poisson distributed nodes (UEs here). This is because  in orthogonal multiple access schemes, like OFDMA, there is  one UE per AP located randomly in the AP's association area that transmits on a given resource block. As a result, the uplink interference  can be viewed as stemming from  a Voronoi perturbed lattice process (see \cite{Soellerhaus} for more discussion), for which an exact   interference characterization is not available.   Moreover, due to the  uplink power control, the transmit power of an interfering UE is correlated with its path loss to the  AP under consideration.  Consequently, various  generative models \cite{NovUL,ElSawyH14,LeeUL14}  have been  proposed to approximate uplink performance in OFDMA Poisson cellular networks.   Most of these models, however, only apply to certain special cases such as (macro-only) for single tier networks \cite{NovUL} or full channel inversion with truncation and nearest AP association \cite{ElSawyH14}.  They do not extend naturally to HCNs with flexible power control and association. The recent work  in \cite{LeeUL14}, however, adopts a similar approach to the one proposed in this paper for  approximating the interfering UE  process  to derive the uplink  $\SIR$ distribution in a two tier network  with a (simpler) linear power control and biased association. All these generative models, however,   ignore the aforementioned conditioning, which may yield unreliable performance estimates. Also, none of these prior works characterizes  the impact of load balancing on the uplink rate distribution or the joint uplink-downlink rate coverage.

\textbf{Joint uplink-downlink coverage.} When UEs employ different association policies for uplink and downlink,  called \emph{decoupled association} \cite{ElshaerBDI14,SmiPopGavTWC,SmiPopGav}),  it   results in possibly different APs serving the user in the uplink and downlink. Characterizing the correlation between the respective   uplink and downlink path losses is then vital for the joint coverage analysis.  Such a correlation analysis was addressed in the recent work  \cite{SmiPopGavTWC} for the  special case of a two-tier scenario with max-received power association for downlink and nearest AP association for uplink.  However, the uplink coverage in \cite{SmiPopGavTWC,SmiPopGav} was derived assuming the interfering user process follows a homogeneous PPP, which is not accurate  for uplink analysis (as discussed above).  The analysis in this paper also addresses the joint uplink-downlink rate and $\SINR$  in greater generality with an arbitrary   association and number of tiers. Traditional coupled association is a special case of this general setting.  

\subsection{Contributions and outcomes}
 A novel generative model is proposed  to analyze uplink performance, where the  APs of each tier are  assumed to  be distributed as an independent homogeneous Poisson point process (PPP) and all UEs employ a weighted path loss based association and FPC.     The  interfering UE locations are approximated   as an \emph{inhomogeneous}  PPP with intensity dependent on the   association parameters. Further, the correlation between the uplink  transmit power of each interfering  UE  and its path loss to the AP under consideration is captured.  Based on this novel approach, the contributions of the paper are as follows:\\
\textbf{Uplink $\SINR$ and rate  distribution.} The complementary cumulative distribution  function  (CCDF) of the  uplink $\SINR$ and rate are  derived for  a  $K$-tier  HCN as a function of the association (tier specific) and power control   parameters in Sec. \ref{sec:ULanalysis}. The general   expression is simplified for certain plausible scenarios.   Simpler  upper and lower bounds are also derived.  \\
\textbf{Joint uplink-downlink rate coverage.} The joint rate/$\SINR$ coverage is defined as the joint probability of uplink and downlink rate/$\SINR$ exceeding  their respective thresholds. The joint coverage is derived in Sec. \ref{sec:ULDLcov} by combining the  derived analysis of  uplink coverage   with the characterization of joint distribution  of  uplink-downlink path losses for  arbitrary    uplink and downlink association weights. The uplink and downlink interference is, however, assumed independent for tractability.

The   analysis  of Sec. \ref{sec:ULanalysis} and  \ref{sec:ULDLcov} (and the involved assumptions)  are validated by comparing with  simulations in  Sec. \ref{sec:validation}  for  a wide  range of parameter settings, which  builds confidence in the   following design insights.\\
\textbf{Insights.} Using the developed model, it is shown, in Sec. \ref{sec:opt}, that:  
\begin{itemize}
\item the  PCF maximizing uplink $\SIR$ coverage  is inversely proportional to the $\SIR$ threshold. As a result, edge users prefer a higher PCF as compared to that of cell interior users. A similar result was shown in \cite{NovUL} for macrocellular networks.

\item With increasing disparity  in association weights across various tiers,  the optimal PCF increases across all $\SIR$ thresholds. 

\item  Minimum path loss  association (i.e. same association weights for all tiers) leads to  optimal uplink  rate coverage. This is in contrast to the corresponding result for downlink \cite{SinDhiAnd13,SinAnd14,AndLoadCommag13}.

\item  For minimum path loss  association and full channel inversion based power control, the uplink $\SIR$ coverage is independent  of infrastructure density in multi-tier networks\footnote{A similar result was shown in \cite{ElSawyH14}   under a different deployment model for interfering UEs.}. This trend is similar to   that in  downlink HCNs \cite{josanxiaand12,dhiganbacand12}. However, the corresponding uplink  $\SIR$ is shown to be stochastically dominated by that of downlink. 

\item With a static uplink-downlink resource allocation ratio, the uplink and downlink association weights that maximize their respective coverage also   maximize the uplink-downlink joint coverage.

\item As a result, a decoupled association---employing different association weights for uplink and downlink---maximizes  joint uplink-downlink rate coverage.

\end{itemize}

\section{System Model}\label{sec:sysmodel}
A co-channel deployment  of a $K$-tier HCN is considered, where the locations of the APs of the $k^{\mathrm{th}}$ tier  are modeled as a 2-D homogeneous PPP   $\PPP_{k} \subset \real{2}$  of density $\dnsty{k}$. All APs  of tier $k$ are assumed to transmit with power $\power{k}$. Further,  the UEs in the network are assumed to be distributed according to an independent homogeneous PPP $\PPPu$ with density $\userdnsty$.    
The signals are assumed to experience path loss with a path loss exponent (PLE) $\ple$ and   the power received from a  node at $X \in \real{2}$ transmitting with power $\power{X}$ at $Y \in \real{2}$  is $\power{X} H_{X,Y}\pl(X,Y)^{-1}$, where $H\in \real{+}$ is the fast fading power gain and $\pl$ is the path loss.  The random channel gains are assumed to be Rayleigh distributed with unit average power, i.e.,  $H \sim \exp(1)$, and $\pl(X,Y) \triangleq \sha_{X,Y} \|X-Y\|^{\ple}$, where $\sha\in \real{+}$ denotes the large scale fading (or shadowing) and   is assumed i.i.d across all UE-AP pairs but the  same for uplink and downlink. The small scale fading gain $\chanl$ is assumed i.i.d  across all links.   WLOG, the analysis in this paper is done for a \textit{typical} UE located at the origin $O$.  The AP serving this typical UE is referred to as the \textit{tagged} AP.

\subsection{Uplink power control}
Let $\cc{X} \in \PPP{}$ denote the AP serving the UE at $X \in \real{2}$ and define $\pl_{X} \triangleq \pl(X,\cc{X})$ to be the path loss  between the UE and its serving AP. A fractional pathloss-inversion based power control  is assumed for   uplink transmission, where  a UE at $X$ transmits with a power spectral density (dBm/Hz) $\power{X}=  \power{u} \pl_{X}^{ \epsilon}$, where $0\leq \epsilon\leq 1$  is the power control fraction (PCF) and $\power{u}$ is the open loop power spectral density \cite{XiaUL}.  Thus, the total transmit power of a user depends on the spectral resources allocated to the user and it's path loss. For  tractability, the  per user maximum power constraint is ignored in this paper.   However, if the dependence of transmit power on load (or resources) is ignored, the analysis in this paper can be extended to incorporate a maximum power constraint similar to  \cite{ElSawyH14}.

 Orthogonal access is assumed in the uplink without multi-user transmission, i.e.,  there is only one UE transmitting in any given resource block. 
  Let $\PPPu^b$ be the point process  denoting the location of UEs transmitting on the same resource as the typical UE. Therefore, $\PPP_{u}^b$ is  \emph{not} a PPP but a Poisson-Voronoi perturbed lattice (per \cite{Soellerhaus}). The uplink $\SINR$ of the typical UE (at $O$) on a given resource block is
\begin{equation}\label{eq:sinr}
\SINR =\frac{ H_{O,\cc{O}} \pl_{O}^{\epsilon-1}}{\SNR^{-1} + \sum_{X\in \PPPu^b}\pl_{X}^{\pcf} H_{X,\cc{O}} \pl(X,\cc{O})^{-1}},
\end{equation}
where $\SNR \triangleq \frac{\power{u}G \mathrm{L_0}}{{\mathrm{N}_0}}$ with $\mathrm{N}_0$ being the thermal noise spectral density, $G$ being the antenna gain at the tagged AP, and $\mathrm{L_0}$ is the free space path loss at a reference distance. 
Henceforth  channel power gain between interfering  UEs and the tagged AP $\{H_{X,\cc{O}}\}$ are simply denoted by $\{H_X\}$.  The index  `$O$' of the typical user is dropped   wherever implicitly clear.

\subsection {Weighted path loss association}
Every UE is assumed to be using weighted path loss for both uplink and downlink  association in  which a UE at $X$ associates to an AP of tier $\ct{X}$ in the uplink, where
\begin{align}\label{eq:association}
\ct{X} &= \arg \max_{k\in\left\{1,\ldots, K \right\}} \metric{k} \pl_{\mathrm{min},k}(X)^{-1},
\end{align}
with $\pl_{\mathrm{min},k}(X)=\min_{Y\in\PPP_{k}}\pl(X, Y)$ is  the minimum path loss of the UE from   $\uth{k}$ tier and $\metric{k}$ is the uplink association weight  for APs in the $\uth{k}$ tier.  The  downlink association is similar with possibly different per tier weights denoted by  $\{\metric{k}^{'}\}_{k=1}^K$ and the selected tier   denoted by $\ct{X}^{'}$. 
 
 The presented association encompasses  biased cell association, where $\metric{k} = \power{k}\bias{k}$ with $\power{k}$ and $\bias{k}$ being the transmit power of APs of $\uth{k}$ tier and the corresponding bias respectively.  Note that if all the association weights are identical, it results in minimum path loss  association.   
For ease of notation, we define  $\nmetric{k} \triangleq \frac{\metric{k}}{\metric{\ct{}}}, \nmetric{k}^{'} \triangleq \frac{\metric{k}^{'}}{\metric{\ct{}^{'}}^{'}} \,\, \forall k=1\ldots K$,   as the ratio of the association weight of an arbitrary tier to that of the serving tier of the typical UE (defined in (\ref{eq:association})) under association weights $\{\metric{k}\}$  and $\{\metric{k}^{'}\}$. 

As a result of the above association model, the uplink association cell of an AP of tier $k$ located at $X$ is
\[ \assocr_{X} =\{ Y \in \real{2}:\metric{k} \pl(X,Y)^{-1} \geq \metric{j}\pl_{\mathrm{min},j}(Y)^{-1},\,\, \forall j=1\ldots K.  \}\]
The downlink association cell can be similarly defined. 
Note that the described association strategy (both for uplink and downlink) is stationary \cite{SinBacAnd13} and hence the resulting association cells are also stationary. The uplink association cells in a two tier setting with $\frac{\power{1}}{\power{2}}=20$ dB resulting from downlink max power association and minimum path loss association are contrasted in Fig. \ref{fig:overview}. 

 It is assumed that each AP has at least one user in its association region with data to transmit in uplink. Further, the AP queues for downlink transmission are assumed to be \textit{saturated} implying that each AP always has data to transmit in downlink. The fraction of resources reserved for the uplink at  each AP is denoted by $\af$.  Assuming an equal partitioning  of the total  uplink (downlink) resources among the associated uplink (downlink)  users (as accomplished by proportional fair or round robin scheduling), the  rate of the typical user is
\begin{equation}\label{eq:ratemodel}
\rate = \frac{\res}{\load{}}\gamma\log\left(1+\SINR\right),
\end{equation}
where $\res$ is the bandwidth,   $\load{}$ denotes the total number of uplink  or downlink users sharing the $\gamma$ fraction of resources, $\gamma = \af$ for uplink and  $1-\af$ for downlink.
 The notation used in this paper  is summarized in Table \ref{tbl:param}.
 \begin{table}
	\centering
\caption{Notation and simulation parameters}
	\label{tbl:param}
  \begin{tabulary}{\columnwidth}{ |C | L | L|}
    \hline
  \textbf{ Notation} & \textbf{Parameter} & \textbf{Value (if applicable) }\\\hline
 $\PPP_k$, $\dnsty{k}$, $\power{k}$ & PPP of tier  $k$ APs, the corresponding density, and the corresponding power   &  \\\hline
 $\PPP_u$, $\userdnsty$ & user PPP and density & $\userdnsty =200$  per sq. km \\\hline
 $\ple $, $\twople$ & path loss exponent; $2/\ple$ & \\ \hline
$\res$ &  bandwidth & $10$ MHz \\  \hline
$\metric{k},\metric{k}^{'}$ & uplink and downlink association weight for tier $k$ &  \\\hline
$\pcf$, $\power{u}$ & power control fraction, open loop power spectral density & $0\leq \pcf\leq 1$, $\power{u}=-80$ dBm/Hz  \\\hline
 $\mathrm{N}_0$  &  thermal noise spectral density & $-174$ dBm/Hz \\\hline
$ \chanl$ & small scale fading gain & Exponential   with unit mean $\sim\exp(1)$\\ \hline
$\sha$ & large scale  fading & Lognormal with $8$ dB standard deviation\\ \hline
$\ct{X}$, $\ct{X}^{'}$ & uplink and downlink serving tier of user at $X$ &  \\ \hline
$\cc{X}$ & serving AP of user at $X$ in uplink & \\\hline
$\load{}$ & uplink or  downlink load  & \\\hline
\end{tabulary}
\end{table}

\begin{figure*}
  \centering
\subfloat[Maximum downlink power association ]{\includegraphics[width= 0.5\columnwidth]{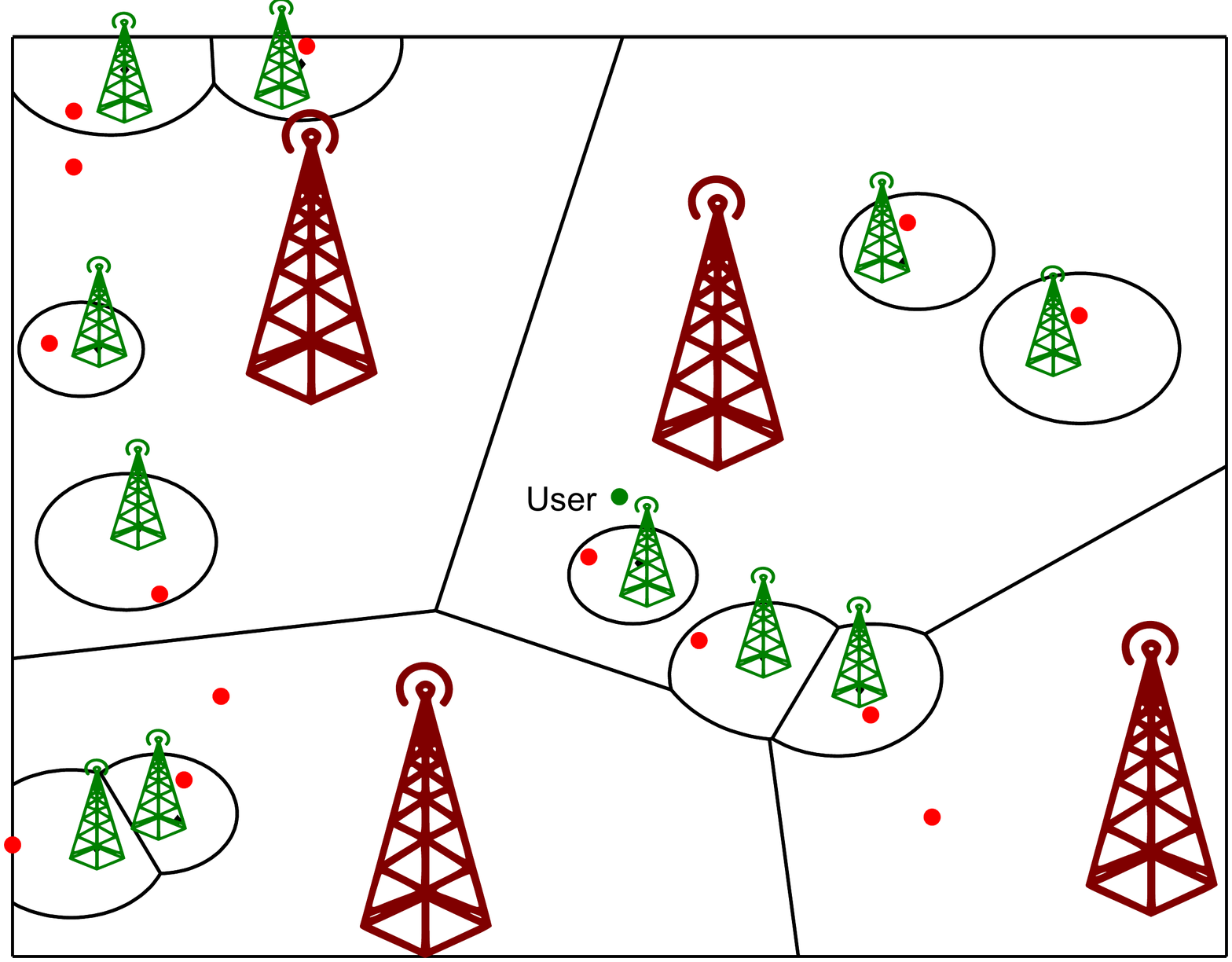}}
\subfloat[Nearest AP association ]{\includegraphics[width= 0.5\columnwidth]{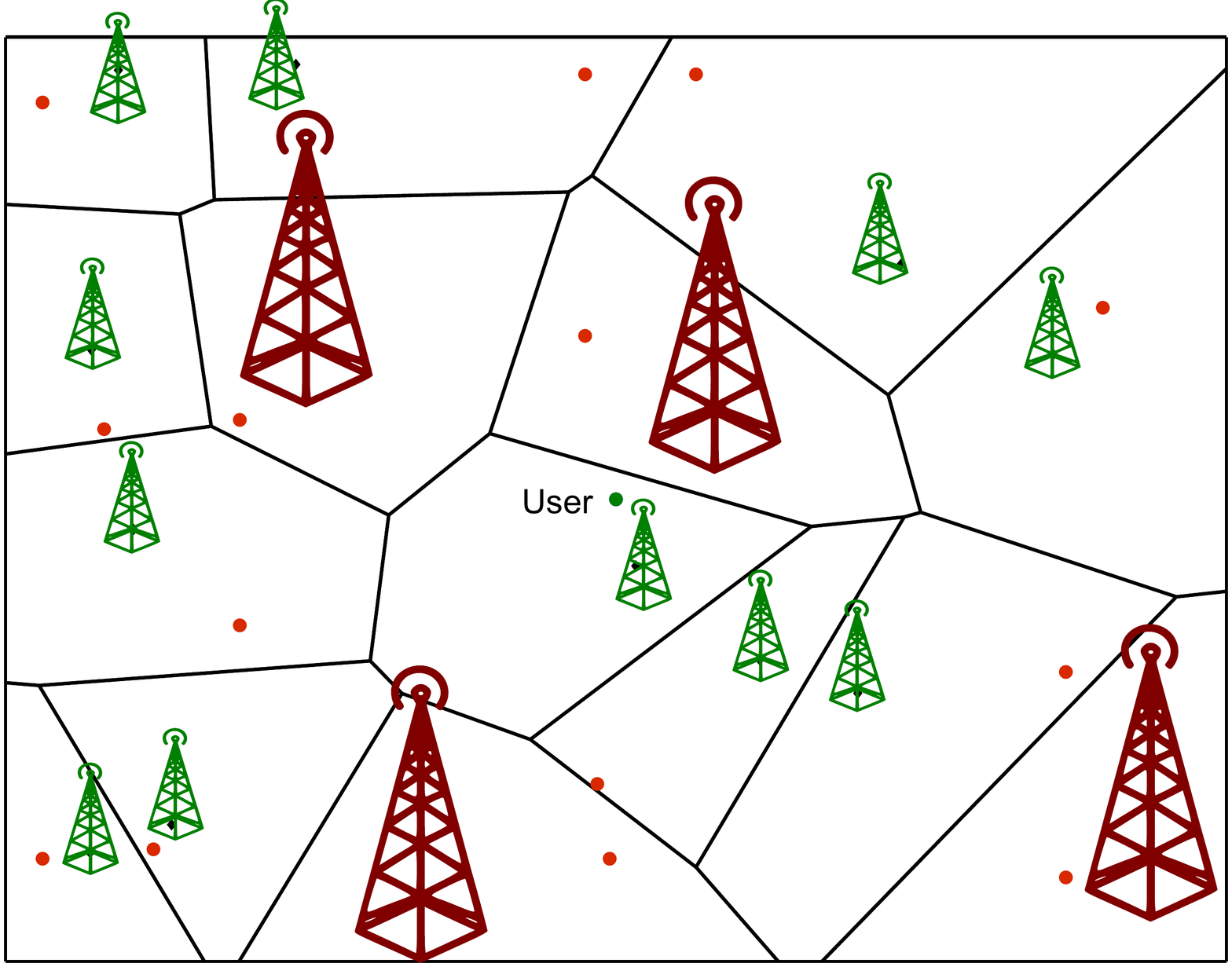}}
\caption{Different association strategies and the corresponding association regions with   one active UE per AP.}
 \label{fig:overview}
\end{figure*} 

\section{Uplink   $\SINR$ and Rate Coverage}\label{sec:ULanalysis}
This is the main technical section of the paper, where we detail the proposed uplink model and the corresponding analysis.  
\subsection{General case}
The uplink $\SIR$ CCDF  of the typical UE is  
\begin{equation}\label{eq:pcovexp}
\pcov(\SINRthresh) \triangleq \pr(\SIR> \SINRthresh) = \sum_{k=1}^K \pr(\ct{}=k)\pcov_k(\SINRthresh),
\end{equation}
where 
\begin{align*}
\pcov_k(\SINRthresh) & \triangleq   \pr(\SIR > \SINRthresh|\ct{}=k)
= \pr\left(\frac{\chanl  \pl{}^{\epsilon-1}}{\sum_{X\in \PPPu^b}\pl_{X}^{\pcf} H_{X}  \pl(X,\cc{})^{-1}}> \SINRthresh|\ct{} =k\right)\\
&= \expect{\exp(-\pl^{1-\epsilon}\SINRthresh I)|\ct{}=k}=\expect{\mgf_{I|\ct{} =k}(\pl^{1-\epsilon}\SINRthresh)},
\end{align*}
where $I=\sum_{X\in \PPPu^b}\pl_{X}^{\pcf} H_{X} \pl(X,\cc{})^{-1}$ is the uplink interference at the tagged AP $\cc{}$, and $\mgf_{I|\ct{}=k}$   is the   Laplace transform of interference conditional on $\uth{k}$ tier being the serving tier.

The following Lemma characterizes the path loss distribution of a typical UE in the given system model.  
\begin{lem}\label{lem:dist}
\textbf{Path loss distribution at the desired link.} The probability distribution function (PDF) of the path loss of a typical UE to its serving AP   is  
\begin{equation*}
f_{\pl}(l) = \twople l^{\twople-1}\sum_{j=1}^Ka_j\exp(-G_j l^{\twople}),  \,\, l \geq 0,
\end{equation*}
where $\twople \triangleq \frac{2}{\ple}$,  $a_k = \dnsty{k}\pi \expect{\sha^{\twople}}$, $ G_k = \sum_{j=1}^K a_{j}(\metric{j}/\metric{k})^{\twople}$,   
and the  PDF, conditioned on the serving the tier being $k$, is
\begin{equation*}
f_{\pl|\ct{}=k}(l)\triangleq f_{\pl}(l|\ct{}  =k ) = \twople G_k l^{\twople-1}\exp(-G_k l^{\twople}),  \,\, l \geq 0,
\end{equation*}
where $\passoc_k \triangleq \pr(\ct{} =k)= \frac{a_{k}}{G_k}$ is the probability of the typical UE  associating with tier $k$.
\end{lem}
\begin{proof}
The proof follows by generalizing the results in \cite{BlaKarKee12,MadHCN12} to our setting.  
Define the propagation process (introduced in \cite{BlaKarKee12}) from   APs  of tier $j$ to the typical user as $\PropPP_j  \triangleq\{ L(X,0)\}_{X\in \PPP_{j}}$.   The process $\PropPP_j$ is also Poisson with intensity $\ndnstyr_j(t) = a_j t^{\twople}, t \in \real{+}$, with $a_j  = \pi\dnsty{j}\expect{\sha^{\twople}}$\footnote{A tier specific $\sha$ can be incorporated in the analysis using $a_j  = \pi\dnsty{j}\expect{\sha_j^{\twople}}$.}. Therefore $ $ $\pr(\pl_{\mathrm{min},j}>t) = \exp(-\ndnstyr_j(t))$. Path loss to the tagged AP of tier $k$ 
has the CCDF
\begin{align*}
 \pr(\pl > l |\ct{} = k) & = \frac{\pr(\pl_{\mathrm{min},k}> l, \ct{}  =k)}{\pr(\ct{}=k) }
\\ &=  \frac{ \pr\left(\cap_{j=1,\neq k}^K\pl_{\mathrm{min},k}^{-1}\metric{k} > \pl_{\mathrm{min},j}^{-1}\metric{j}\cap \pl_{\mathrm{min},k}> l \right)}{\passoc_k}\\
& =  \frac{\cexpect{\prod_{j=1,\neq k}^K \exp(-\ndnstyr_j(\nmetric{j}\pl_{\mathrm{min},k}))}{\pl_{\mathrm{min},k}>l}}{\passoc_k} \\ &= \frac{\twople a_k}{\passoc_k} \int_{l}^\infty y^{\twople-1}\exp\left(-y^{\twople}\sum_{j=1}^K a_j \nmetric{j}^{\twople}\right)\mathrm{d} y
\end{align*}
Therefore,  $f_{\pl}(l|\ct{} =k) = \twople \frac{\ednsty{k}}{\passoc_k} l^{\twople-1}\exp(-G_k l^\twople)$ and $f_{\pl}(l)=\sum_{j=1}^K \passoc_j f_{\pl}(l|\ct{} =j)$.
\end{proof}
The above distribution is not, however, identical to the distribution of the path loss between an interfering UE and its serving AP, since the latter is the conditional distribution given that the interfering UE \emph{does not} associate with the tagged AP. This correlation is formalized in the corollary below.
\begin{cor}\label{cor:cdist}
\textbf{Path loss distribution at an interfering UE.} The PDF of the path loss of a  UE at $X$ associated with tier $j$, conditioned on it not lying in the association cell ($\assocr_{\cc{}}$) of the tagged AP at $\cc{} $   of tier $k$ and the corresponding path loss $\pl(X,\cc{} )=y$, is 
\begin{equation*}
f_{\pl_X}(l|\ct{X} =j, \ct{} =k, X \notin \assocr_{\cc{}}, \pl(X,\cc{} )=y  ) \\= \frac{\twople G_j}{1-\exp(- G_k y^{\twople})}l^{\twople-1}\exp(- G_j l^{\twople}), \,\, 0\leq l\leq  \frac{\metric{j}}{\metric{k}}y.
\end{equation*}
\end{cor}
\begin{proof}
An interfering UE at $X$ cannot associate with the tagged AP of tier $k$ which, given the association policy, implies that the corresponding path loss   is bounded as $\pl_X \leq \frac{\metric{j}}{\metric{k}}\pl(X,\cc{})$. Noting that $G_j\left(\frac{\metric{j}}{\metric{k}}\right)^{\twople}= G_k$ results in the above distribution. 
\end{proof}

Due to uplink orthogonal access within each AP, only one UE per AP  transmits on the typical resource block and hence contributes to interference at the tagged AP. Therefore, $\PPP_{u}^b$ is  \emph{not} a PPP but a Poisson-Voronoi perturbed lattice (per \cite{Soellerhaus}) and hence the functional form of the interference (or the Laplace functional of  $\PPP_{u}^b$) is not tractable. Based on the  following remark, we propose an approximation to characterize the corresponding process as an \emph{inhomogeneous} PPP.
%
%
%

\begin{rem}\label{rem:thin}
{\textbf{Thinning probability.}}
Conditioned on an  AP of tier $k$ being located at $V \in \real{2}$, a UE at $U \in \real{2}$  associates  with $V$ with probability $\pr(\cc{U}=V)=\exp(-G_k \pl(V,U)^{\twople})$. 
\end{rem}

\begin{asmptn}{\textbf{Proposed interfering UE point process.}}\label{asmptn:ppp}
Conditioned on the tagged AP being located at $\cc{ }$ and of tier $k$, the propagation process  of  interfering UEs from tier $j$  to $\cc{ }$, $\PropPP_{u,j}\coloneqq \{\pl(X,\cc{})\}_{X \in \PPP_{u,j}^b }$ is assumed to be Poisson with intensity measure function $\ndnstyr_{u,j}(dx)= \twople a_j x^{\twople-1} (1-\exp(-  G_k x^{\twople}))(dx)$.  
\end{asmptn}
The basis of the above assumption is Remark \ref{rem:thin} along with   the fact that only one UE per AP can potentially interfere with the typical UE in the uplink.  Thus,  the maximum density of UEs that might potentially interfere in the uplink  from tier $j$ is $\dnsty{j}$. Assuming this parent process to be a PPP with density $\dnsty{j}$, the propagation process of these UEs to the tagged AP has intensity measure function $ \twople a_j  x^{\twople-1}$. However, the intensity of this parent process has to be appropriately thinned  as per Remark 1 to account for  the fact that these UEs do not associate with the tagged AP. The resulting process $\PropPP_{u,j}$ has an intensity that increases with increasing path loss from the tagged AP. 

 The methodology proposed in \cite{LeeUL14} for modeling non-uniform intensity of $\PPP_{u,b}$ was based on a curve-fitting based approach  and hence may not be accurate for more diverse system parameters.

\begin{asmptn}{\textbf{Tier-wise independence.}}\label{asmptn:indpndnc}
The point process of interfering UEs from each tier are assumed to be independent, i.e., the intensity measure of the interfering UEs propagation process  $\PropPP_u$ is $\ndnstyr_u(x)\triangleq  \sum_{j=1}^K \ndnstyr_{u,j}(x)$.
\end{asmptn}

\begin{asmptn}{\textbf{Independent path loss.}}\label{asmptn:dist}
The path losses $\{\pl_{X}\}_{X \in \PPP_{u}^b}$ are assumed to follow the Gamma distribution given by Corollary \ref{cor:cdist},  assumed independent (but not identically distributed).
\end{asmptn}

\begin{lem}\label{lem:LapI}
The Laplace transform of  interference at the tagged AP of tier $k$ under the proposed model 
is 
\begin{equation}\label{eq:LapIk}
\mgf_{I_k}(s)\triangleq \mgf_{I|\ct{}=k}(s) = \exp\left(-\frac{\twople}{1-\twople}s\sum_{j=1}^K \nmetric{j}^{1-\twople}a_{j}\cexpect{\pl^{\twople-(1-\pcf)}\hgc_\twople\left( \frac{s\nmetric{j}}{\pl^{1-\pcf}}\right)}{\pl|\ct{}=j}\right),
\end{equation}
 where $\hgc_\twople(x)\triangleq\,  \hg(1,1-\twople,2-\twople,-x)$ and   $\hg$ is the Gauss-Hypergeometric function.
\end{lem}
\begin{proof}
See Appendix \ref{sec:prooflapI}.
\end{proof}

Using the above Lemma and (\ref{eq:pcovexp}), the uplink $\SINR$ coverage is given in the following Theorem.
\begin{thm}\label{thm:pcov} The uplink $\SINR$  coverage probability for the proposed uplink generative model is
\small 
\begin{align*}
 \sum_{k=1}^{K} \twople a_k\int_{l>0}l^{\twople-1}\exp\left(-G_k l^{\twople} -\frac{\twople}{1-\twople}\SINRthresh l^{1-\pcf}\sum_{j=1}^K \left(\frac{\metric{j}}{\metric{k}}\right)^{1-\twople}a_{j}\cexpect{\pl^{\twople-(1-\pcf)}\hgc_\twople\left(\frac{\SINRthresh\metric{j}l^{1-\pcf}}{\metric{k}\pl^{1-\pcf}}\right)}{\pl|\ct{}={j}}-\frac{\SINRthresh}{\SNR} l^{1-\pcf} \right)\mathrm{d}l.
\end{align*}
\normalsize
\end{thm}
The $\SIR$ coverage can be derived by letting $\SNR\to \infty$ in the above theorem. 
\begin{cor} \label{cor:sir} The uplink $\SIR$  coverage probability for the proposed uplink generative model is 
\small
\begin{align*}
\pcov(\SINRthresh)= \sum_{k=1}^{K} \twople a_k\int_{l>0}l^{\twople-1}\exp\left(-G_k l^{\twople} -\frac{\twople}{1-\twople}\SINRthresh l^{1-\pcf}\sum_{j=1}^K \left(\frac{\metric{j}}{\metric{k}}\right)^{1-\twople}a_{j}\cexpect{\pl^{\twople-(1-\pcf)}\hgc_\twople\left(\frac{\SINRthresh\metric{j}l^{1-\pcf}}{\metric{k}\pl^{1-\pcf}}\right)}{\pl|\ct{}={j}}\right)\mathrm{d}l.
\end{align*}
\normalsize
\end{cor}

The coverage expression for the most general case involves two folds of integrals and a lookup table for the Hypergeometric function. The expression is, however, further simplified for the special cases   in the next section. Useful bounds can hence be obtained in single integral-form (Corollary \ref{cor:ub}) and closed-form (Corollary \ref{cor:lb}) as below: 

\begin{cor}\label{cor:ub}The uplink $\SIR$ coverage for the proposed   generative model is  upper bounded by 
\small
\begin{align*}
 \pcov^u(\SINRthresh)=\sum_{k=1}^{K} \twople a_k\int_{l>0}l^{\twople-1}\exp\left(-G_k l^{\twople}  -\frac{\twople\SINRthresh l^{1-\pcf}}{(1-\twople)\Gamma(2+(1-\pcf)/\twople)} \sum_{j=1}^K \left(\frac{\metric{j}}{\metric{k}}\right)^{1-\twople}a_{j}G_j^{(1-\pcf)/\twople-1} \hgc_\twople\left(\frac{\SINRthresh\metric{j}l^{1-\pcf}G_j^{(1-\pcf)/\twople}}{\metric{k}\Gamma(2+(1-\pcf)/\twople)}\right) \right)\mathrm{d}l.
\end{align*}
\normalsize
\end{cor}
\begin{IEEEproof}
See Appendix \ref{sec:proofub}.
\end{IEEEproof}

\begin{rem} It can be noted from the above proof that the coverage upper bound is, in fact, exact for full channel inversion, i.e., $\pcf=1$.
\end{rem}

\begin{cor}\label{cor:lb}
The uplink $\SIR$ coverage is lower bounded by   
\begin{equation*} 
\pcov^l(\SINRthresh)= \exp\left(-\SINRthresh^{\twople}\frac{\pi^2\twople \pcf(1-\pcf)}{\sin(\pi\twople) \sin(\pi \pcf)}\left(\sum_{k=1}^K\frac{\ednsty{k}}{G_k^{2-\pcf}}\right)\left(\sum_{k=1}^K\frac{\ednsty{k}}{G_k^{\pcf}}\right)\right).
\end{equation*}
\end{cor}
\begin{proof}
See Appendix \ref{sec:prooflb}.
\end{proof}

\subsection{Special cases}
For  the following plausible special  cases, the uplink  $\SIR$ coverage expression is further simplified.
\begin{cor}(\textbf{$K=1$}) The uplink $\SIR$ coverage in a single tier network with density $\dnsty{1}$ is 
\begin{equation*}
\pcov(\SINRthresh)=  \twople \ednsty{1}\int_{l>0}l^{\twople-1}\exp\left(-\ednsty{1} l^{\twople} -\frac{\twople}{1-\twople}\SINRthresh l^{1-\pcf} \ednsty{1}\cexpect{\pl ^{\twople-(1-\pcf)}\hgc_\twople\left(\frac{\SINRthresh l^{1-\pcf}}{\pl ^{1-\pcf}}\right)}{\pl}\right)\mathrm{d}l,
\end{equation*}
where $\ednsty{1}=  \dnsty{1}\pi \expect{\sha^{\twople}}$.
\end{cor}
The above expression differs from the one in \cite{NovUL} due to the proposed  interference characterization. In \cite{NovUL}, the distribution of path loss of each interfering UE to its serving AP was assumed i.i.d. 

\begin{cor}($\metric{j} = \metric{k} \forall j,k$) The uplink $\SIR$ coverage in a $K$-tier network with min-path loss  association is the same as the coverage of a single tier network with density $\dnsty{}=\sum_{k=1}^K \dnsty{k}$.
\end{cor}

\begin{cor} ($\pcf=0$)
Without uplink  power control,  the uplink $\SIR$ coverage is  
\begin{align*} 
\pcov(\SINRthresh)= \sum_{k=1}^{K} \twople a_k\int_{l>0}l^{\twople-1}\exp\left(-G_k l^{\twople}-\ednsty{}\int_{0}^\infty \frac{1-\exp(- G_k  x)}{1+(\SINRthresh l)^{-1}x^{-1/\twople}}\mathrm{d}x\right)\mathrm{d}l,
\end{align*}
where $a =\sum_{j=1}^K a_j$.
\end{cor}

\begin{cor} ($\pcf=1$)
With full  channel inversion,  the coverage is 
\begin{align*}
\pcov(\SINRthresh)= \sum_{k=1}^K \frac{\ednsty{k}}{G_k}\exp\left(-\frac{\twople}{\twople-1} \SINRthresh\sum_{j=1}^{K} \left(\frac{\metric{j}}{\metric{k}}\right)^{1-\twople}\frac{\ednsty{j}}{G_j}\hgc_\twople\left( \SINRthresh \frac{\metric{j}}{\metric{k}}\right)\right) .
\end{align*}
\end{cor}

\begin{cor} ($\pcf=0,\metric{j} = \metric{k} \forall j,k$)
Without power control and with min path loss association, the uplink $\SIR$ coverage is  
\begin{align*}
\pcov(\SINRthresh)=  \twople \ednsty{}\int_{l>0}l^{\twople-1}\exp\left(-\ednsty{ } l^{\twople}-\ednsty{ }\frac{\twople}{1-\twople}\SINRthresh l  \cexpect{\pl ^{\twople-1}\hgc_\twople\left(\frac{\SINRthresh l}{\pl }\right)}{\pl}\right)\mathrm{d}l.\nonumber
\end{align*}
\end{cor}

\begin{cor}\label{cor:pcovpcmin} ($\pcf=1, \metric{j} = \metric{k} \forall j,k$)
With full channel inversion based power control and with min path loss association, the uplink $\SIR$ coverage is 
\begin{align*}
\pcov(\SINRthresh)= \exp\left(-\frac{\twople\SINRthresh}{1-\twople} \hgc_\twople(\SINRthresh)\right) .
\end{align*}
\end{cor}

\begin{rem}\textbf{Comparison with downlink.} The  downlink $\SIR$ coverage derived in   \cite{josanxiaand12}   for max downlink received power association, when adapted to the current setting,  is $\frac{1}{1+\frac{\twople\SINRthresh}{1-\twople} \hgc_\twople(\SINRthresh)}$. Since  ${1+x} < \exp(x),\,\, \forall x >0$,   downlink $\SIR$ stochastically dominates the uplink $\SIR$ of Corollary \ref{cor:pcovpcmin}.
\end{rem}

\begin{rem}\textbf{Density invariance.}
Corollary \ref{cor:pcovpcmin}  highlights  the independence of uplink $\SIR$ coverage on  infrastructure density in HCNs with minimum path loss association  and full channel inversion. This trend  is similar to the result proved for downlink $\SIR$ in macrocellular networks \cite{andganbac11} and HCNs \cite{dhiganbacand12,josanxiaand12}. 
\end{rem}

\subsection{Uplink rate distribution}
The rate of a user depends on both the $\SINR$ and load at the tagged AP (as per (\ref{eq:ratemodel})), which in turn depends on the  corresponding association area $|\assocr_{\cc{}}|$. The weighted path loss association  and PPP placement of APs leads to  complex association cells (see Fig. \ref{fig:overview}) whose area distribution is not known. However, the association policy   is  stationary   \cite{SinBacAnd13} and hence the mean uplink association area  of  a typical AP of tier $k$ is $\frac{\passoc_k}{\dnsty{k}}$.  The association area approximation  proposed in  \cite{SinDhiAnd13}  is used to quantify the uplink load distribution at the tagged AP as
\[\loadpmf_t(\userdnsty \passoc_k, \dnsty{k}, n) \triangleq \pr(N =n |\ct{}=k) = \frac{3.5^{3.5}}{(n-1)!}\frac{\Gamma(n+3.5)}{\Gamma(3.5)}\left(\frac{\userdnsty\passoc_{k}}{\dnsty{k}}\right)^{n-1}\left(3.5 + \frac{\userdnsty\passoc_{k}}{\dnsty{k}}\right)^{-(n+3.5)}  \,\,,n\geq 1.  \]

Using Corollary \ref{cor:sir} and (\ref{eq:ratemodel}), and assuming the independence between $\SINR$ and load, the uplink rate coverage is given  in the following Theorem.
\begin{thm}\label{thm:rcov} Under the presented system model and assumptions, the uplink rate coverage is given by 
\[\rcov(\RATEthresh) = \sum_{k=1}^K\passoc_k\sum_{n>0}\loadpmf_t(\userdnsty \passoc_k, \dnsty{k}, n)\pcov_k(2^{\nRATEthresh n}-1),\]
where $\pcov_k$ is given in Corollary \ref{cor:sir} and $\nRATEthresh \triangleq \RATEthresh(\af\res)^{-1}$. 
\end{thm}
\begin{proof} Using the rate expression in (\ref{eq:ratemodel}) 
\begin{align*} 
\pr(\rate >\RATEthresh) &= \pr(\SINR>2^{\nRATEthresh\load{}}-1) =\sum_{k=1}^K\passoc_k\pr(\SINR>2^{\nRATEthresh\load{}}-1|\ct{}=k)\\
& =\sum_{k=1}^K\passoc_k\sum_{n>0}\loadpmf_t(\userdnsty \passoc_k, \dnsty{k}, n)\pr\left(\SINR>2^{\nRATEthresh n}-1|\ct{}=k, \load{}  = n\right),
\end{align*}
where $\nRATEthresh = \RATEthresh(\af\res)^{-1}$ is the normalized rate threshold.
Since  APs with larger association regions have higher load and larger user to AP distance, therefore  the load and $\SINR$ are correlated.  For tractability, this dependence and thermal noise are ignored, as in \cite{SinDhiAnd13}, to yield \[
  \pr\left(\SINR>2^{\nRATEthresh n}-1|\ct{}=k, \load{} = n\right)  \approx \pcov_k(2^{\nRATEthresh n}-1).\]
\end{proof}
\begin{cor}\label{cor:mrcov}
If the load at each AP is approximated by its respective mean, $\avload{k} \triangleq \expect{N|\ct{}=k}= 1+1.28\frac{\passoc_k\dnsty{u}}{\dnsty{k}}$ \cite{SinDhiAnd13}, the uplink rate coverage is
\[\avrcov(\RATEthresh) = \sum_{k=1}^K\passoc_k\pcov_k(2^{\nRATEthresh \avload{k}}-1).\]
\end{cor}
The corollary above simplifies the rate coverage expression of Theorem \ref{thm:rcov} by eliminating a sum and sacrificing a bit of accuracy. 

\section{Joint uplink-downlink rate coverage}\label{sec:ULDLcov}
The joint uplink-downlink rate coverage is defined formally below.
\begin{definition}
The uplink-downlink joint rate coverage is the probability that the rate  on both  links exceed  their respective thresholds, i.e.,
\begin{align*}
\rcov^J(\RATEthresh_u,\RATEthresh_d ) & \triangleq \pr(\text{Uplink rate }>\RATEthresh_u, \text { Downlink rate }>\RATEthresh_d). \end{align*}
\end{definition}
It can be equivalently interpreted as the fraction of users in the network whose  both uplink and downlink rate exceed their respective thresholds.

For deriving the joint coverage,  the joint path loss distribution needs to be characterized.  
For the special case  of coupled association, the path losses are identical, however, for the general case they are correlated. The following Lemma characterizes the joint distribution of path losses for arbitrary downlink and uplink association weights.
\begin{lem}\label{lem:jointpl}
{\bf Joint path loss distribution.}
The joint PDF of uplink path loss ($\pl$) and downlink path loss ($\pl^{'}$) for the typical user under the given setting is
\begin{multline*}
f_{\pl,\pl^{'}}(x,y,\ct{}=k,\ct{}^{'}=j)\\=
\begin{cases}
\ednsty{j}\ednsty{k}\twople^2 x^{\twople-1}y^{\twople-1}\exp\left(-\sum_{i=1}^K\ednsty{i} \max\left(\frac{\metric{i}^{'}}{\metric{j}^{'}}y,\frac{\metric{i}}{\metric{k}}x\right)^\twople\right), &  k \neq j,\,\,  \frac{\metric{j}}{\metric{k}}\leq \frac{\metric{j}^{'}}{\metric{k}^{'}},\,\,\, x \geq 0, \,\,  \frac{\metric{j}}{\metric{k}}\leq \frac{y}{x}\leq \frac{\metric{j}^{'}}{\metric{k}^{'}}\\
\ednsty{k}\twople  x^{\twople-1}\exp\left(-x^\twople\sum_{i=1}^K\ednsty{i} \max\left(\frac{\metric{i}^{'}}{\metric{k}^{'}},\frac{\metric{i}}{\metric{k}}\right)^\twople \right), & k = j,\,\, x \geq 0, \,\, y = x\\
0 &  \text{ otherwise.}
\end{cases}\normalsize
\end{multline*}
\end{lem}
\begin{proof}  For $\ct{}=\ct{}^{'}=k$,
\begin{align*}
\pr(\pl=\pl^{'}>x, \ct{}=\ct{}^{'}=k) &= \pr\left(\bigcap_{i=1,\neq k}^K \pl_{\mathrm{min},i} > \pl_{\mathrm{min},k} \max\left(\frac{\metric{i}^{'}}{\metric{k}^{'}},\frac{\metric{i}}{\metric{k}}\right)\bigcap\pl_{\mathrm{min},k} > x \right)\\
& =\twople \ednsty{k} \int_{x}^\infty  l^{\twople-1} \exp\left(-l^\twople\sum_{i=1}^K \ednsty{i} \max\left(\frac{\metric{i}^{'}}{\metric{k}^{'}},\frac{\metric{i}}{\metric{k}}\right)^\twople\right)\mathrm{d}l.
\end{align*}
And for $\ct{}=k,\ct{}^{'}=j$ with $k\neq j$ and $\frac{\metric{j}^{'}}{\metric{k}^{'}}> \frac{\metric{j}}{\metric{k}}$
\small
\begin{multline*}
\pr(\pl > x,\pl^{'}>y, \ct{}=k,\ct{}^{'}=j) \\= \pr\left(\bigcap_{i=1,\neq k,j}^K \pl_{\mathrm{min},i} >  \max\left(\frac{\metric{i}^{'}}{\metric{j}^{'}}\pl_{\mathrm{min},j},\frac{\metric{i}}{\metric{k}}\pl_{\mathrm{min},k}\right)\bigcap\left\{ \frac{\metric{j}}{\metric{k}}\pl_{\mathrm{min},k}\leq \pl_{\mathrm{min},j}\leq \frac{\metric{j}^{'}}{\metric{k}^{'}}\pl_{\mathrm{min},k} \right\}\bigcap\{\pl_{\mathrm{min},k} > x\}, \bigcap\{\pl_{\mathrm{min},j} > y \} \right)\\
=\twople^2 \ednsty{j}\ednsty{k}\int_{x}^\infty \int_{\max(\metric{j}/\metric{k}t,y)}^{\metric{j}^{'}/\metric{k}^{'}t}  t^{\twople-1}u^{\twople-1}\exp\left(-\sum_{i=1}^K\ednsty{i} \max\left(\frac{\metric{i}^{'}}{\metric{j}^{'}}t,\frac{\metric{i}}{\metric{k}}u\right)^\twople\right)\mathrm{d}u\mathrm{d}t
\end{multline*}
\normalsize
Differentiating the above CCDFs leads to the corresponding PDFs. 
\end{proof}

The  downlink $\SIR$ analysis in \cite{josanxiaand12} ignored shadowing. However the analysis  can be adapted to the presented setting to give the Laplace transform of the downlink interference in the following Lemma (presented without proof).
\begin{lem}\label{lem:dlcov}  The Laplace transform of downlink interference ($I^{'}$) when the serving (downlink) AP belongs to tier $j$ and the corresponding path loss is $l$, i.e. $\pl^{'}=l$, is 
\[\mgf_{I^{'}_{j}}(s|\pl^{'}=l)= \exp\left\{-\frac{\twople}{1-\twople}s l^{\twople-1}\sum_{i=1}^K\ednsty{i}\power{i} \nmetric{i}^{'(\twople-1)}\hgc_\twople(s\power{i}/l\nmetric{i}^{'})\right\}.  \]
\end{lem}

\begin{thm}\label{thm:jointrcov}
Using the mean load approximation for uplink and downlink, and assuming the uplink and downlink interference to be independent,  the  joint uplink-downlink rate coverage is
\[\rcov^J(\RATEthresh_u,\RATEthresh_d) =  \sum_{k=1}^K\sum_{j=1}^J\int_{0}^\infty\int_{0}^\infty\mgf_{I_k}\left(x^{1-\pcf}\uRATEthresh(\nRATEthresh_u \avload{k})\right)\mgf_{I'_j}\left(y\power{j}\uRATEthresh(\nRATEthresh_d \avload{k}^{'})\right)f_{\pl,\pl^{'}}(x,y,\ct{}=k,\ct{}^{'}=j)\mathrm{d}x\mathrm{d}y,\]
where $\avload{k}=1+1.28\frac{\dnsty{u}}{\dnsty{k}} \frac{\ednsty{k}}{ \sum_{j=1}^K a_{j}\nmetric{j}^{\twople}}$ is the average uplink load at the AP serving in the uplink,  $\avload{k}^{'}=1+1.28\frac{\dnsty{u}}{\dnsty{k}} \frac{\ednsty{k}}{ \sum_{j=1}^K a_{j}\nmetric{j}^{{'}\twople}}$  is the average downlink load at the  AP serving in the downlink, $\nRATEthresh_u =\RATEthresh (\res\af)^{-1}$, $\nRATEthresh_d =\RATEthresh (\res(1-\af))^{-1}$, and $\mgf_{I_k},\mgf_{I'_j} $  are as defined in Lemma \ref{lem:LapI} and  \ref{lem:dlcov} respectively.
\end{thm}
\begin{proof}
The proof follows by noting that the joint rate coverage can be written in terms of joint $\SIR$ coverage   as in Theorem \ref{thm:rcov}. Assuming independence of uplink and downlink interference, the joint $\SIR$ coverage  is 
\[\pr(\text{ Uplink } \SIR > \SINRthresh_u, \text{Downlink } \SIR> \SINRthresh_d)= \sum_{k=1}^K\sum_{j=1}^J\expect{\indic(\ct{}=k)\indic(\ct{}^{'}=j)\mgf_{I_k}\left(\pl^{1-\pcf}\SINRthresh_u\right)\mgf_{I'_j}\left(\pl^{'}\power{j}\SINRthresh_d\right)}.\] 
The final expression is then obtained by using the rate model of (\ref{eq:ratemodel}) along with  Lemma   \ref{lem:jointpl}.
\end{proof}

\section{Validation through simulations} \label{sec:validation}
The proposed model and the corresponding analytical results are  validated by simulations (based on the model of Sec. \ref{sec:sysmodel}) in a two tier setting with $\dnsty{1}=5$ BS per sq. km, $\ple=3.5$,  $\sha$ assumed Lognormal with $8$ dB standard deviation, open loop power spectral density $\power{u} =-80$ dBm/Hz, and free space path loss $\mathrm{L_0} = -40$ dB is assumed at reference distance of $1$ m at carrier frequency of $2$ GHz   (the same parameters are used in the later sections unless otherwise specified).  Fig. \ref{fig:uplinkcomp} shows the uplink $\SINR$ distribution obtained from simulations along with the $\SIR$ distribution  from analysis (Corollary \ref{cor:sir}) for different association weights and tier 2 densities. Two simplifications of the analysis A1, A2 are also shown for the $\pcf=1$ case. A1   neglects the conditioning of the  transmit power of an interfering user, i.e. Lemma \ref{lem:dist} is used for path loss distribution of interfering UEs instead of Corollary \ref{cor:cdist}. A2 additionally neglects the proposed modeling of interfering UE  process as per Assumption \ref{asmptn:ppp} and instead models interfering UEs of each tier  as an  homogeneous PPP with density same as the corresponding tier density.  The plots lead to the following takeaways: \begin{inparaenum} \item the proposed analysis matches the simulations quite closely for a range of parameters, validating Assumptions  \ref{asmptn:ppp},  \ref{asmptn:indpndnc}, and \ref{asmptn:dist}; \item  neglecting   the proposed thinning and/or conditioning (as is done prior works) leads to significant diversion from actual coverage, and \item   thermal noise has a  minimal impact on uplink $\SINR$ (this could also be due to the higher BS density). \end{inparaenum} Note that a value of $\nmetric{2} = -20 $ dB corresponds to  a typical power difference between small cells and  macrocells and hence is equivalent  to   downlink maximum power association. 

The rate coverage obtained from simulation and analysis (Corollary \ref{cor:mrcov}) is compared for a two-tier setting in Fig. \ref{fig:ratecompa} and for a three-tier setting in Fig. \ref{fig:ratecompb}. The user density used in these plots is $\dnsty{u} = 200$ per sq. km.  The joint rate distribution derived from analysis and simulation is shown in Fig. \ref{fig:jointratecomp} for an uplink resource fraction $\af=0.5$. The close match between analysis and simulations for a wide range of parameters  in these plots  validates the mean load assumption and  the downlink-uplink interference independence assumption.
\begin{figure}
  \centering
\subfloat[$\dnsty{2}=6\dnsty{1}$, $ \metric{2}/\metric{1}=-20$ dB ]{\includegraphics[width= 0.5\columnwidth]{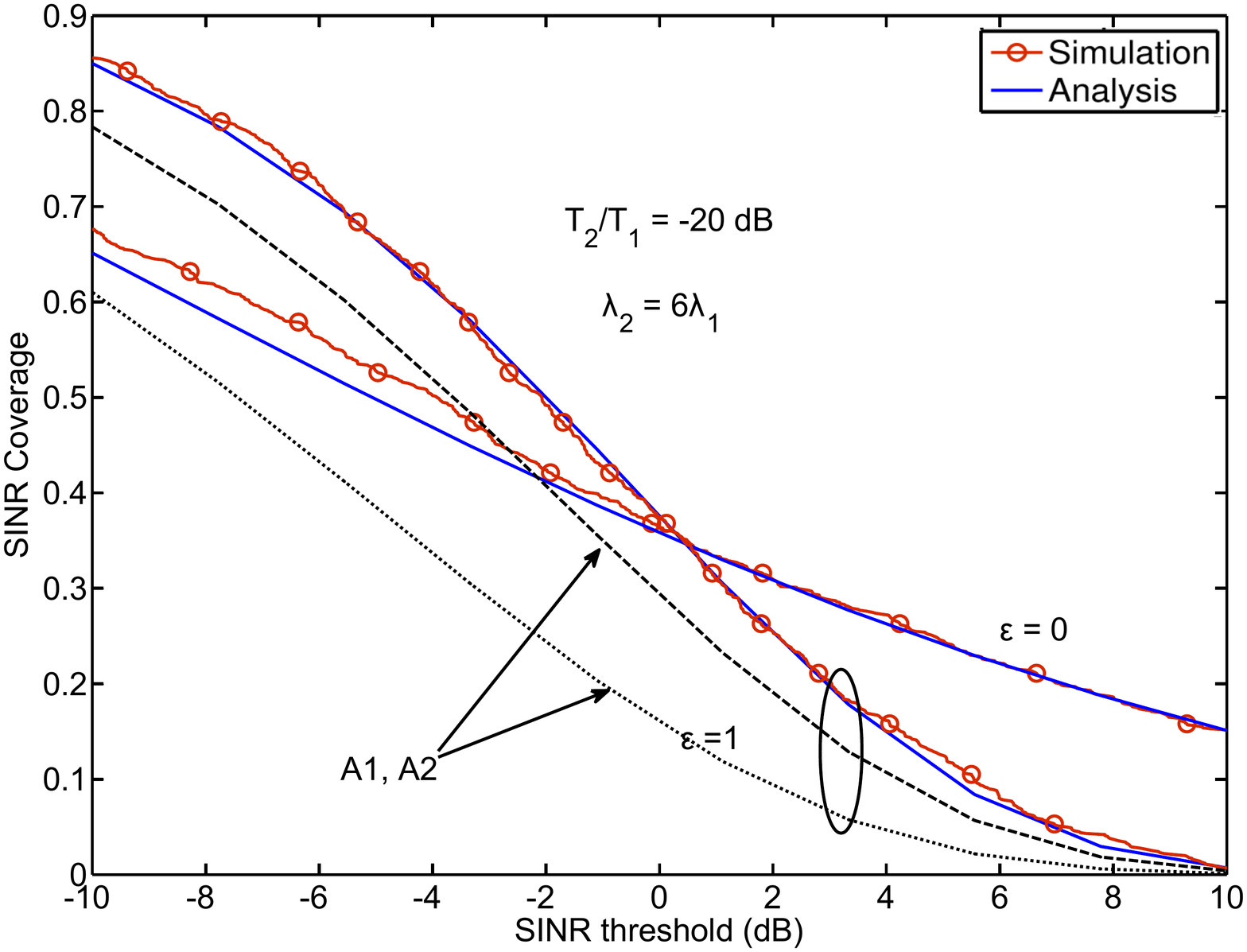}}
\subfloat[ $\dnsty{2}=4\dnsty{1}$, $ \metric{2}/\metric{1}=0$ dB]{\includegraphics[width=0.5 \columnwidth]{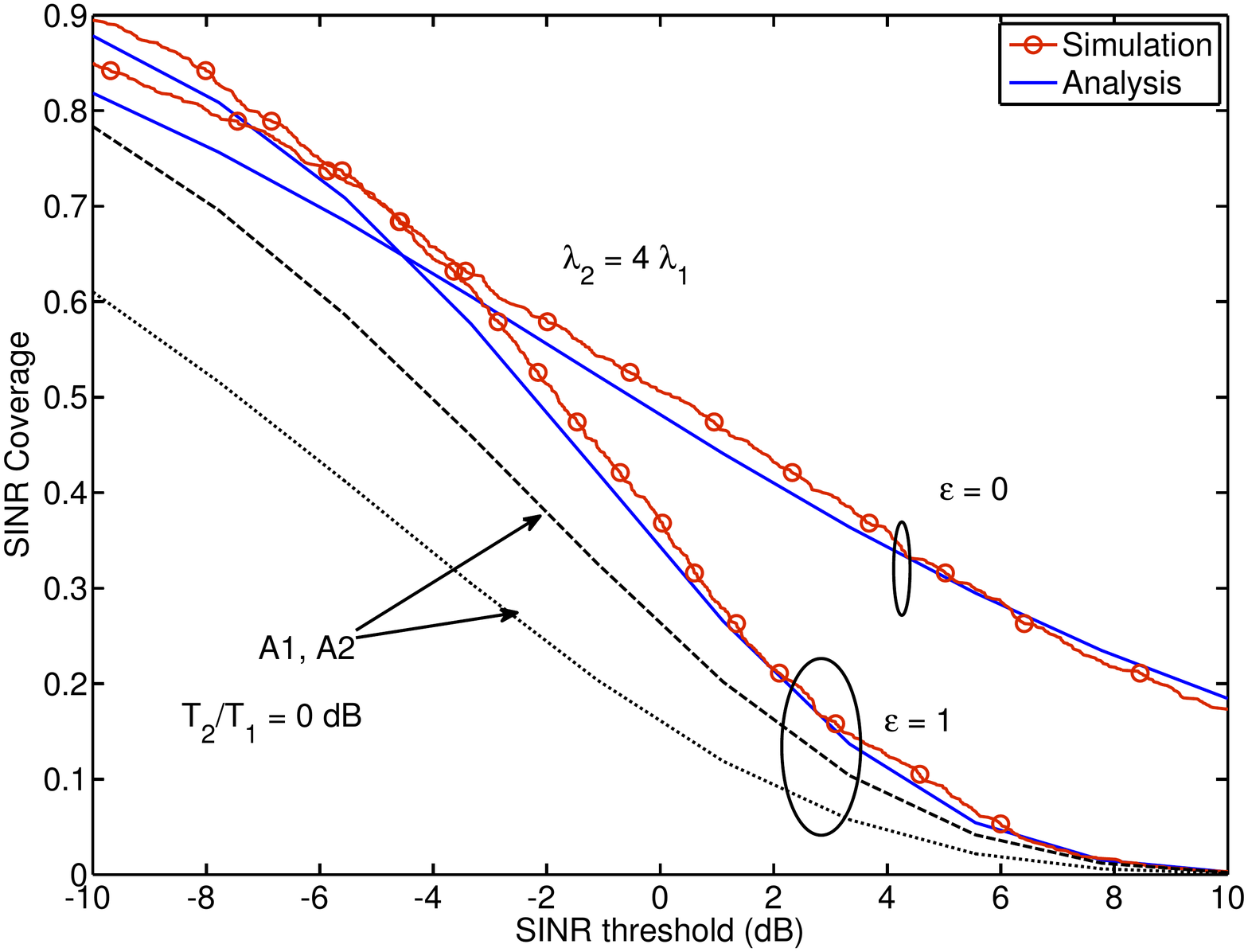}}
\caption{Comparison of uplink $\SINR$ distribution from simulation with $\SIR$ distribution   from analysis.}
 \label{fig:uplinkcomp}
\end{figure} 

 \begin{figure}
  \centering
\subfloat[$\dnsty{2}=6\dnsty{1}$, $ \metric{2}/\metric{1}=-20$ dB ]{\includegraphics[width= 0.5\columnwidth]{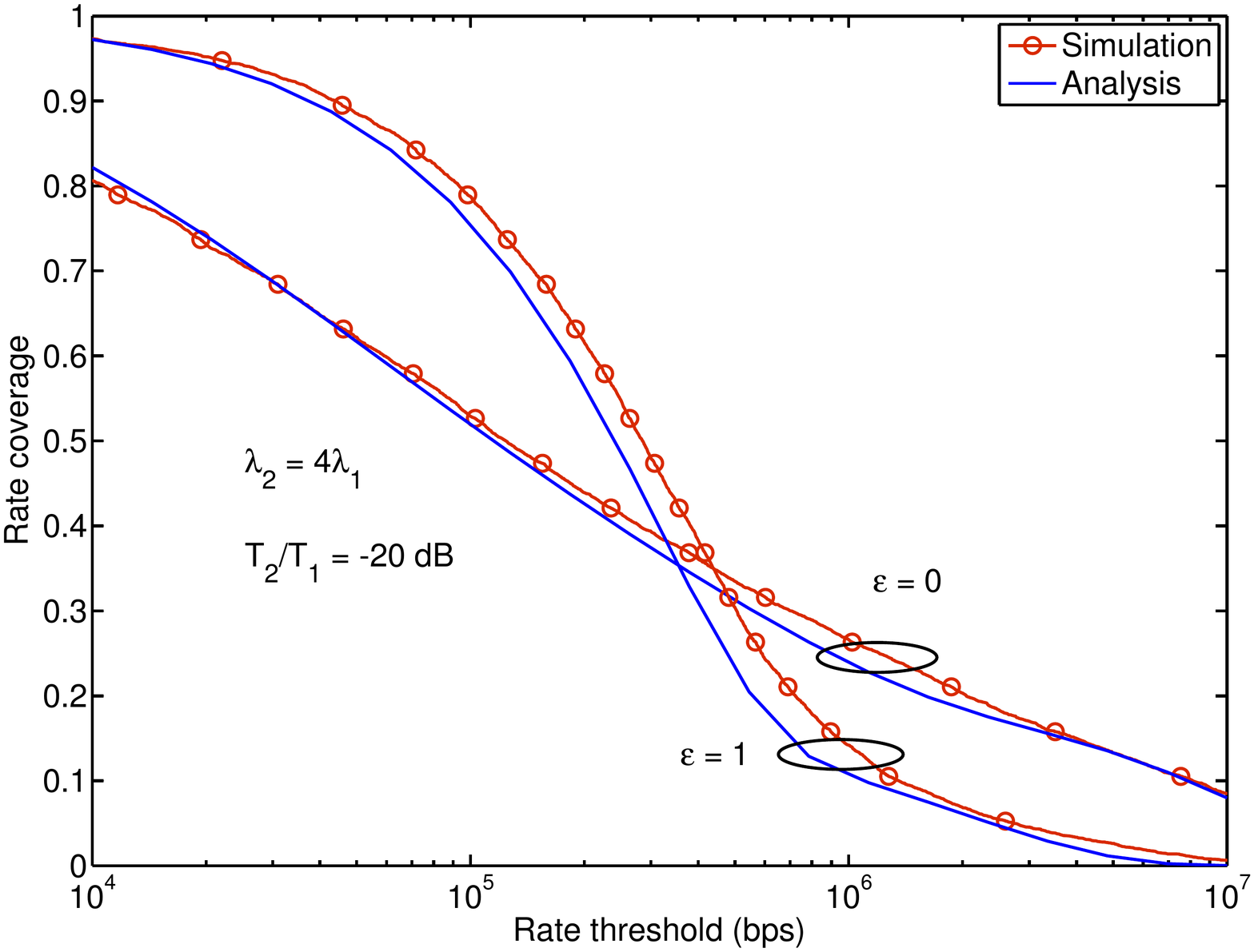}\label{fig:ratecompa}}
\subfloat[$\dnsty{3}=7\dnsty{1}$, $\dnsty{2}=4\dnsty{1}$, $\metric{3}= \metric{2}=\metric{1}$  ]{\includegraphics[width= 0.5\columnwidth]{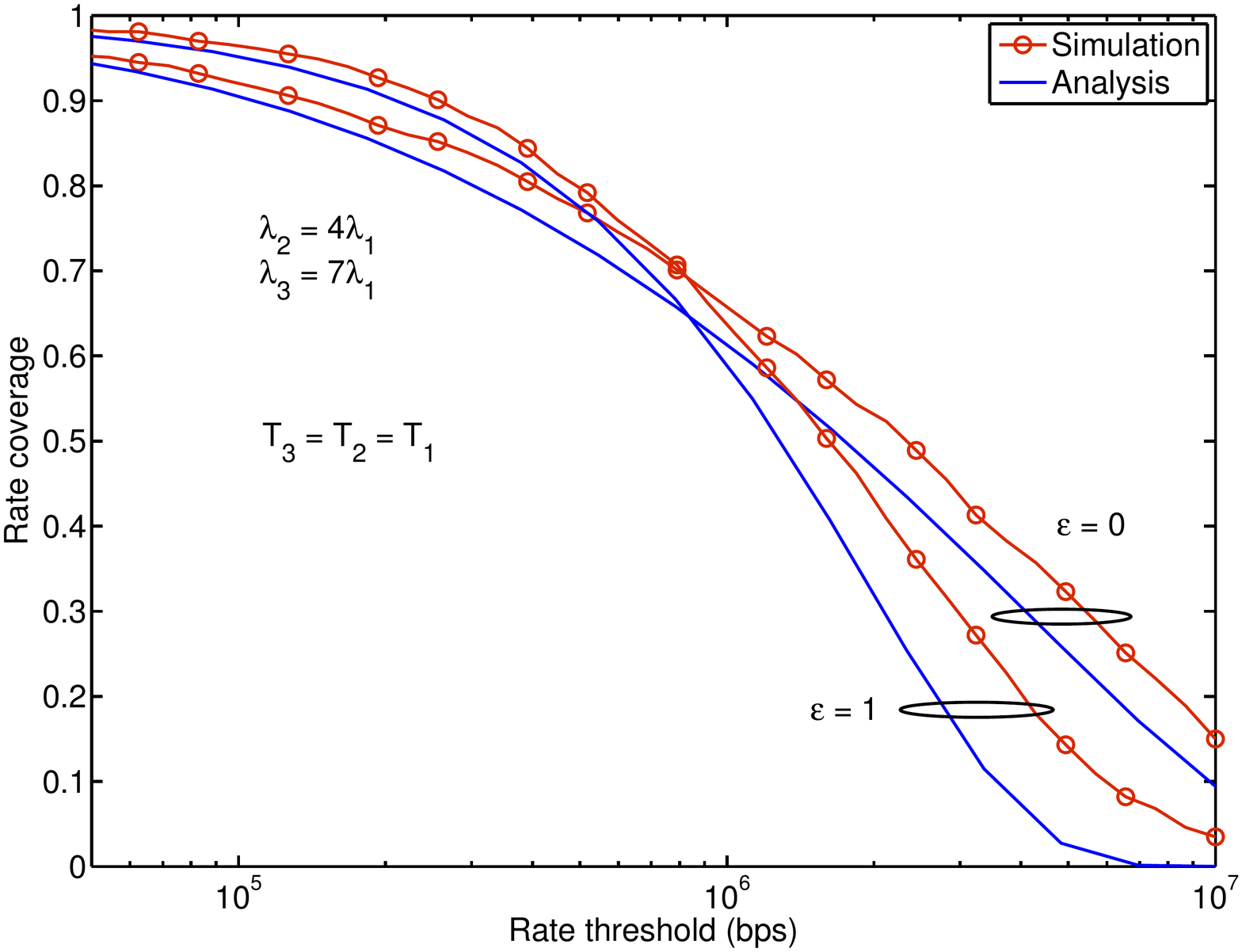}\label{fig:ratecompb}}
\caption{Comparison of uplink rate distribution from analysis and simulation.}
 \label{fig:ratecomp}
\end{figure} 

\begin{figure}
  \centering
{\includegraphics[width= 0.5\columnwidth]{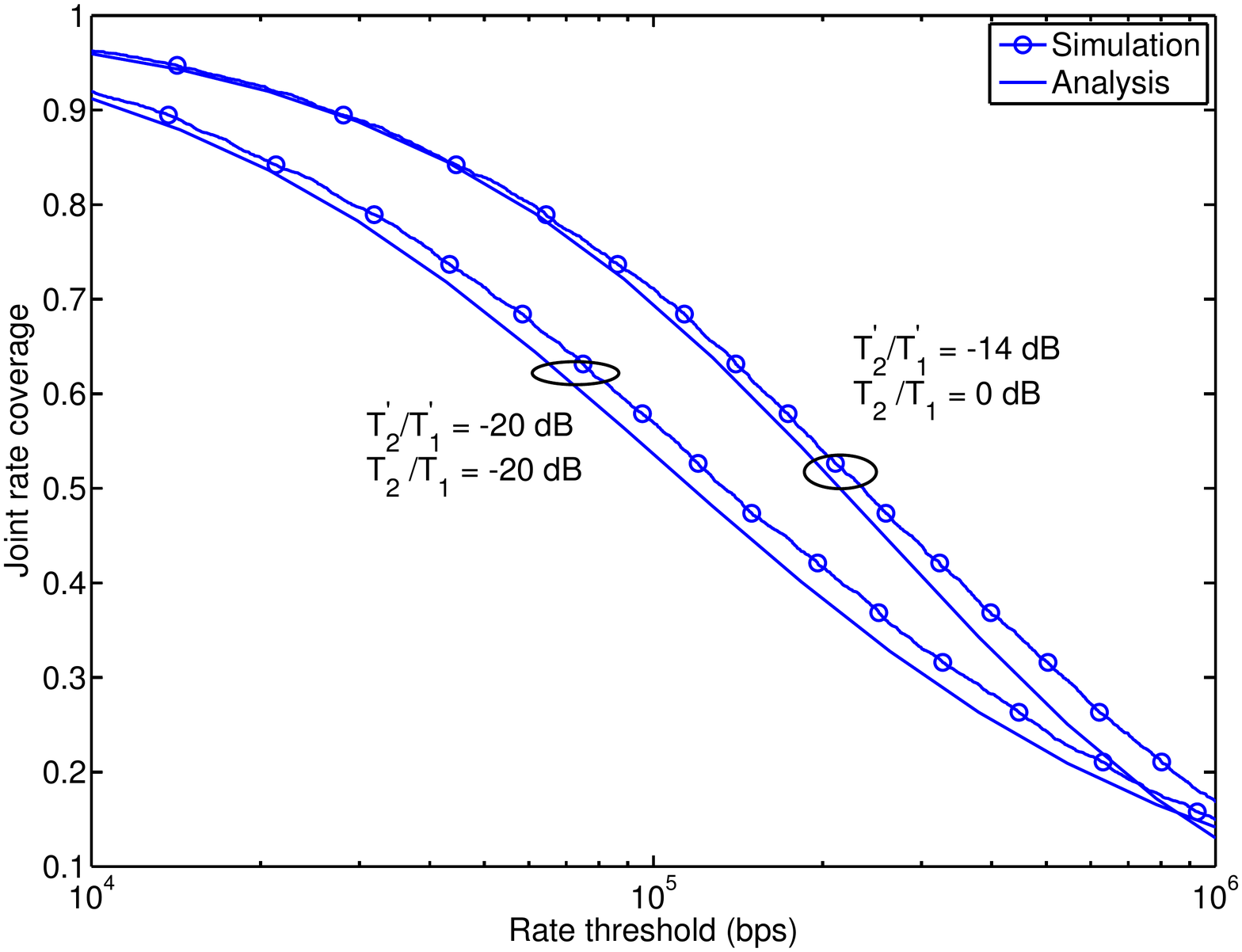}}
\caption{Comparison of joint rate coverage from analysis and simulation (with $\dnsty{2}=6\dnsty{1}$ and $\pcf=0.5$).}
 \label{fig:jointratecomp}
\end{figure}

\section{Optimal power control and association}\label{sec:opt}
The uplink $\SIR$ and rate coverage probability expressions of Corollary \ref{cor:sir}, Theorem \ref{thm:rcov}, and \ref{thm:jointrcov} can be used to numerically find the optimal power control and association weights. However, first   we focus on the coverage lower bound $\mathcal{P}^l$ of Corollary \ref{cor:lb} and obtain the following proposition.
\begin{prop}\label{prop:nn}
Minimum path loss association maximizes $\pcov^l$  $\forall \pcf \in [0,1] \,\,\forall \SINRthresh $. Further, $\pcf=0.5 $ maximizes the coverage lower bound $\pcov^l$.
\end{prop}
\begin{proof}
Using Corollary \ref{cor:lb}, $\pcov^l$ is maximized with  $\{\metric{j}^*\}$ given by
\begin{align*}
\{\metric{j}^*\}_{j=1}^K & =  \argmin \sum_{k=1}^K\frac{\ednsty{k}}{G_k^{2-\pcf}} \sum_{k=1}^K\frac{\ednsty{k}}{G_k^{\pcf}} 
= \argmin \frac{(\sum_{k=1}^K\ednsty{k}\metric{k}^{2-\pcf})(\sum_{k=1}^K\ednsty{k}\metric{k}^{\pcf})}{(\sum_{j=1}^K \ednsty{j}\metric{j})^{2}} \\
& = 1+ \argmin \frac{\sum_{i\neq j} \ednsty{i}\ednsty{j}(\metric{i}^{2-\pcf}\metric{j}^{\pcf}-\metric{i}\metric{j})}{(\sum_{j=1}^K \ednsty{j}\metric{j})^{2}},
\end{align*} where the last equation is minimized  with $\metric{j}=\metric{k}$ $\forall j,k$.  
Moreover, for such a case
\[\pcov^l(\SINRthresh) = \exp\left(-\SINRthresh^{\twople}\frac{\pi^2\twople\pcf}{\sin(\pi\twople)\sin(\pi\pcf)}\right),\]
which is maximized for $\pcf=0.5$.
\end{proof}
\begin{rem} Since the lower bound overestimates the uplink interference by neglecting the correlation of the transmit power of an interfering user with its path loss  to the tagged AP  (and hence treating it as if  originating from an ad-hoc network), the  result of optimal PCF of $0.5$ is in agreement  with   results for ad hoc wireless networks  \cite{JinUL,BacLi11} (derived under quite different modeling assumptions, though). \end{rem}

\textbf{Power control.} 
Since the power control impacts only uplink $\SIR$ and not load (unlike association), the optimal PCF is obtained using the $\SIR$ coverage of Corollary  \ref{cor:sir}. The $\SIR$ threshold  plays a vital role in determining the optimal PCF. More channel inversion is more beneficial for cell edge UEs, as they suffer from higher path loss and as a result the optimal PCF  decreases with $\SIR$ threshold,  as shown in Fig. \ref{fig:OptPCFContourDens}. This is similar to the insight obtained for single-tier networks    in \cite{NovUL}. It is interesting to note that the result on optimal PCF  of Proposition  \ref{prop:nn}  applies only to moderate $\SIR$ thresholds.  Further, as can be observed a higher association weight imbalance leads to a uniform (across all thresholds) increase in the optimal PCF,  as the path losses  in the network increase.  It can also be observed that the optimal PCF is relatively insensitive to  different densities in the two tier network, with no dependence seen in the case of minimum path loss association. A similar trend translates to uplink rate distribution too. The variation of uplink fifth percentile rate (or edge rate, $\RATEthresh| \rcov(\RATEthresh) = 0.95$) and median rate  ($\RATEthresh| \rcov(\RATEthresh) = 0.50$) with PCF is shown in Fig. \ref{fig:ratepcf}.  A higher   PCF maximizes fifth percentile rate than that for median rate, since former  represents users with lower uplink $\SIR$.

 \begin{figure}
  \centering
{\includegraphics[width=0.5 \columnwidth]{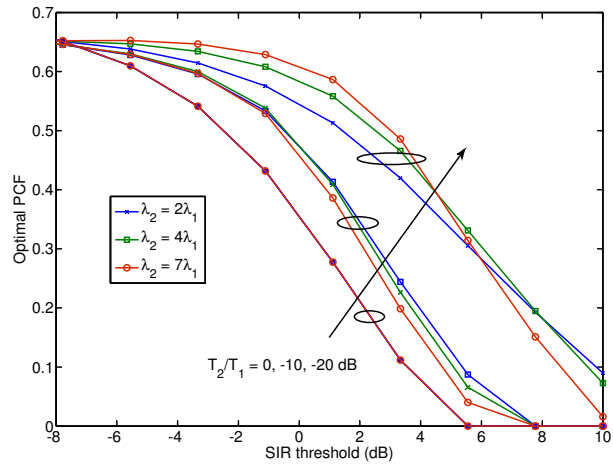}}
\caption{Variation of the optimal PCF  with $\SIR$ threshold (obtained from Corollary \ref{cor:sir}) for various association weights and densities. The bottom ($0$ dB) curve is indistinguishable for all $\dnsty{2}$. }
 \label{fig:OptPCFContourDens}
\end{figure} 
 \begin{figure}
  \centering
\subfloat[]{\includegraphics[width= 0.5\columnwidth]{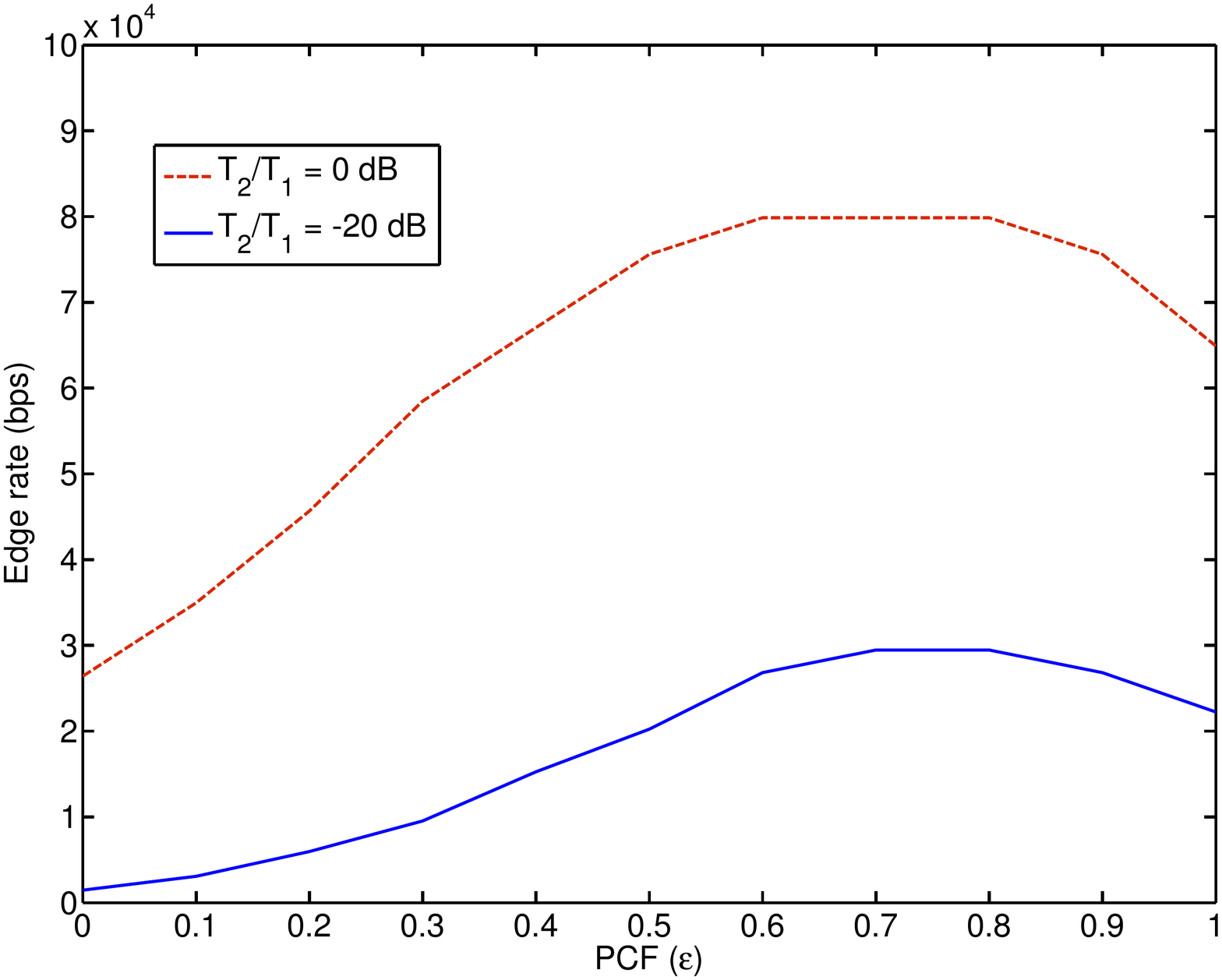}}
\subfloat[]{\includegraphics[width= 0.5\columnwidth]{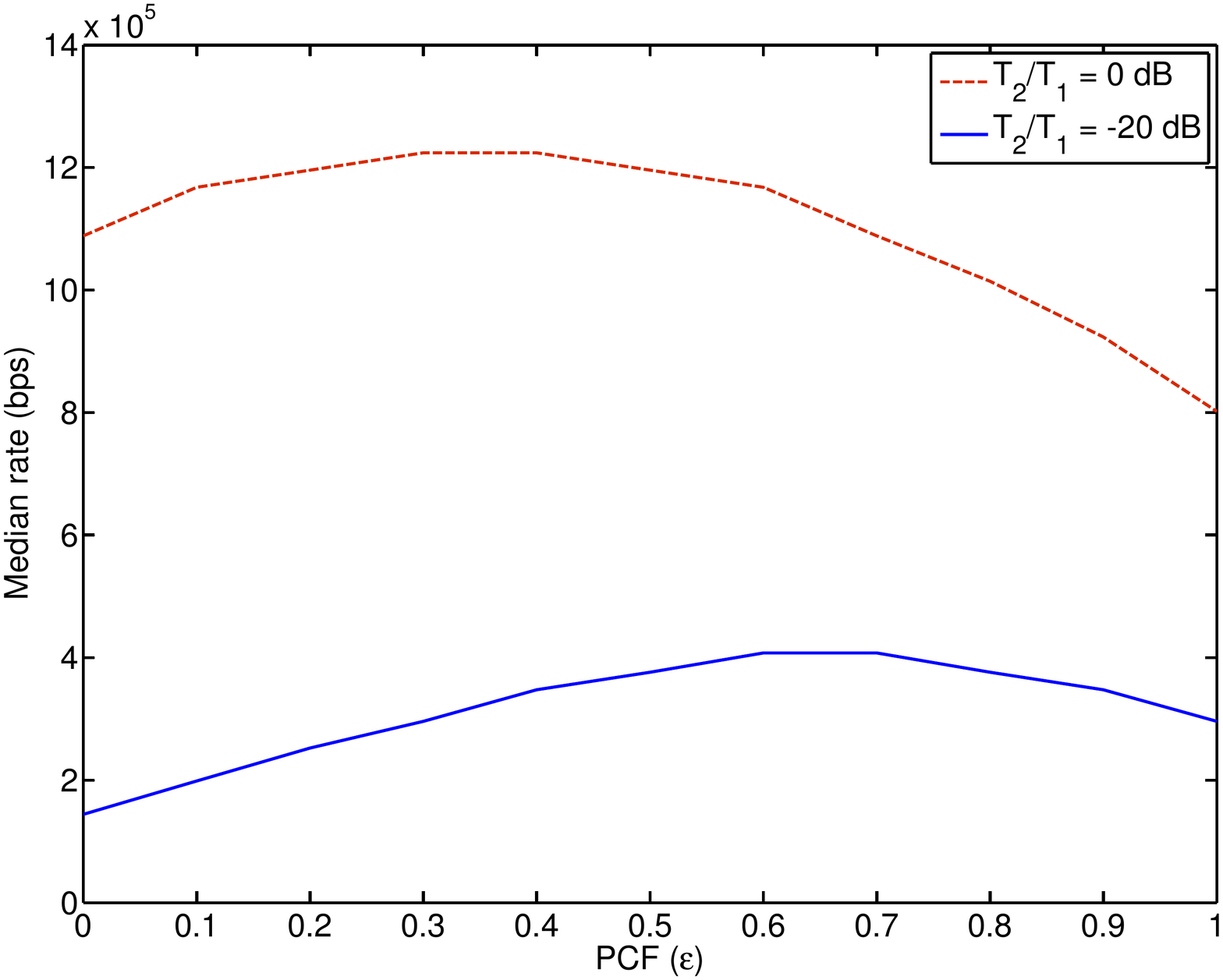}}
\caption{Variation  of uplink edge and median rate with PCF (obtained from Corollary \ref{cor:mrcov}) for $\dnsty{2}=6\dnsty{1}$ per sq. km.}
 \label{fig:ratepcf}
\end{figure}

\begin{figure}
  \centering
{\includegraphics[width= 0.5\columnwidth]{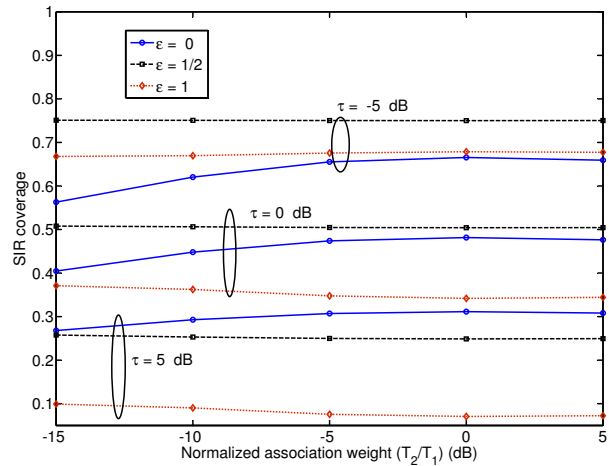}}
\caption{Uplink $\SIR$ variation with association weights (obtained from Corollary \ref{cor:sir}) with $\dnsty{2}=5\dnsty{1}$   for different thresholds and PCFs.}
 \label{fig:CovWeight}
\end{figure}

\textbf{Uplink association weights.} 
The variation of  uplink  $\SIR$ coverage with association weights is shown in Fig. \ref{fig:CovWeight} for different PCFs  and $\SIR$ thresholds.  Association weights are  seen to affect the $\SIR$ coverage nominally, except for the no power control case (where the variation is in concurrence with the result of the Proposition \ref{prop:nn}). An intuitive explanation of this behavior is as follows:  higher weight imbalance may lead a user to associate with a farther macrocell with a higher path loss, but it  would also experience reduced uplink interference due to the larger association area of the corresponding AP. So these two contrary effects compensate for each other leading to the observed phenomenon.

It is worth noting here  that minimum path loss association  leads to identical load distribution across all APs and hence balances the load. Moreover due to no adverse effect on uplink $\SIR$, minimum path loss association is also seen to be optimal from rate perspective too. The trend of uplink edge (fifth percentile) and median rate with association weights is shown in Fig. \ref{fig:ratebias}. As can be seen, irrespective of the PCF and density,  minimum path loss  association is optimal for uplink rate. Note that these results and insights for uplink are in  contrast with the corresponding result for downlink, where maximum $\SIR$ association (equivalent to maximum downlink received  power association) is optimal  for downlink $\SIR$ coverage \cite{josanxiaand12}, and  hence a conservative association bias\footnote{A bias of $\sim 6$ dB  was shown to maximize edge and median rates in downlink \cite{SinAnd14,ye2012user}   with $20$ dB power difference between macro and small cell, which translates to $\metric{2}^{'}/\metric{1}^{'}=-14$ dB for the setting of this paper.} was   shown to be optimal for rate coverage \cite{SinAnd14, ye2012user}.

 \begin{figure}
  \centering
\subfloat[Edge rate, $\pcf=1$]{\includegraphics[width= 0.5\columnwidth]{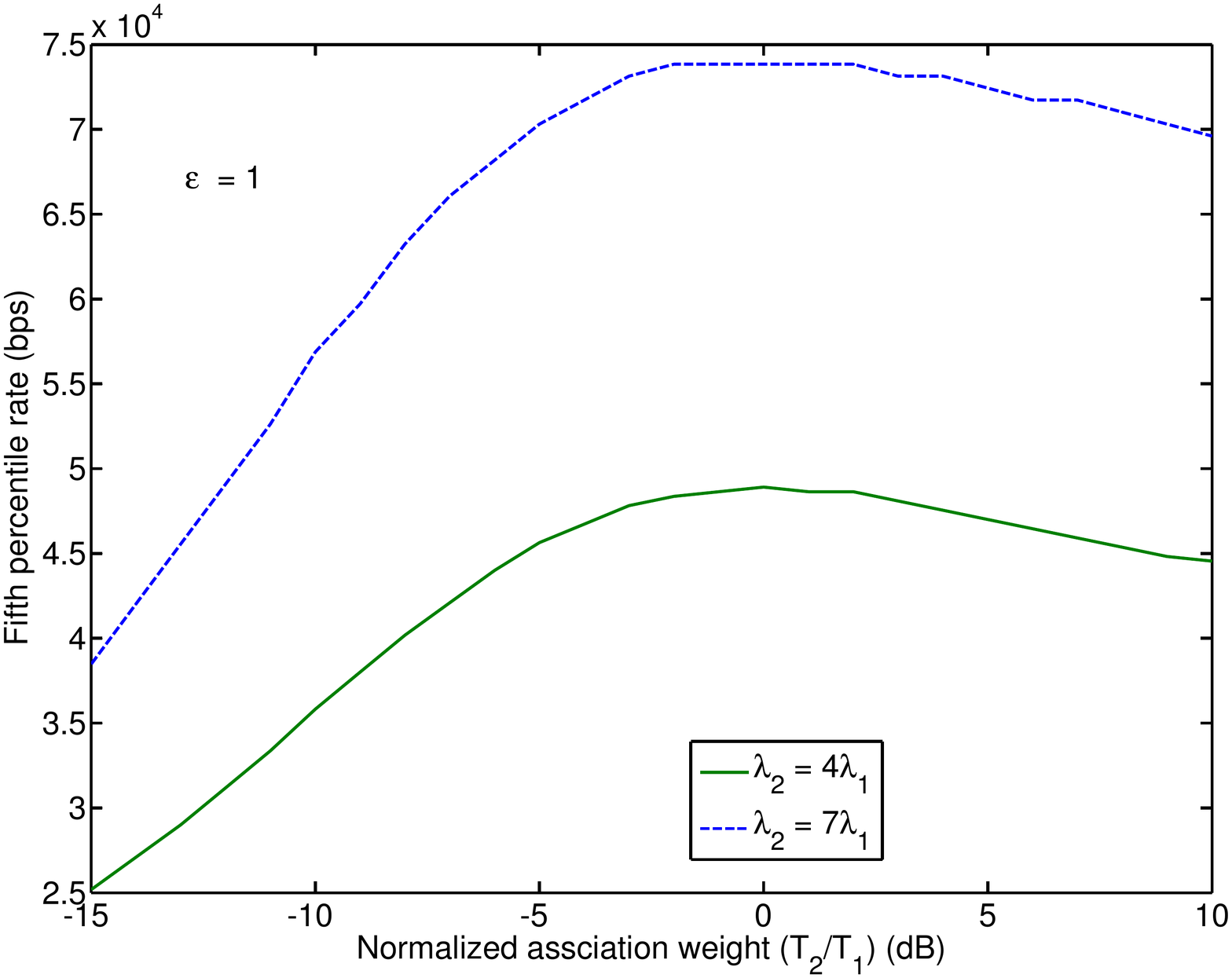}}
\subfloat[Median rate, $\pcf=0$]{\includegraphics[width= 0.5\columnwidth]{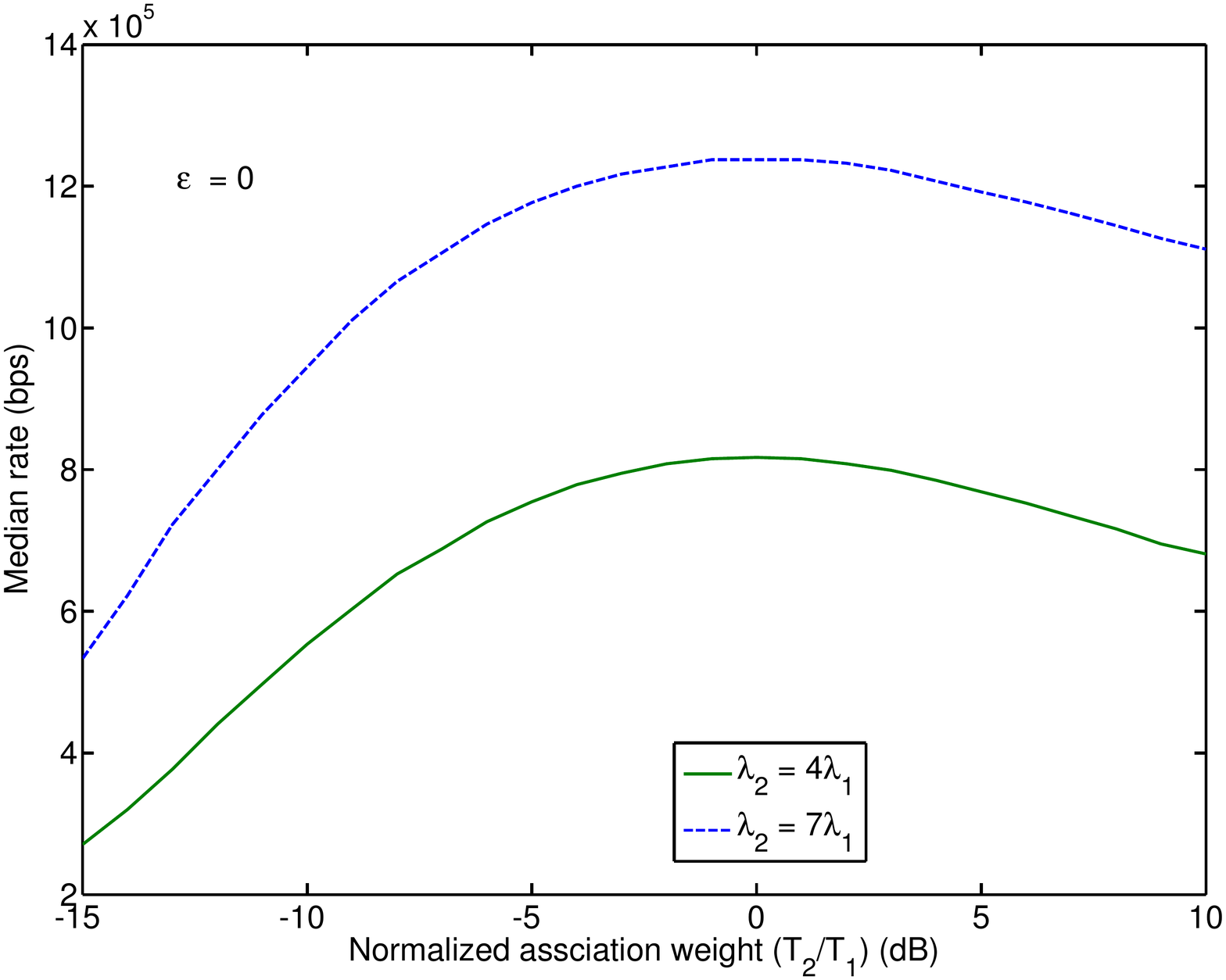}}
\caption{Variation  of uplink edge and median rate  with association weights (obtained from Corollary \ref{cor:mrcov}).}
 \label{fig:ratebias}
\end{figure} 

\textbf{Uplink-downlink jointly optimal association.} Considered separately, as discussed in the previous section, the association weights $\frac{\metric{2}}{\metric{1}} = 0 $ dB, $\frac{\metric{2}^{'}}{\metric{1}^{'}} = -14 $ dB optimize the uplink and downlink rate respectively. However, what happens if the joint downlink and uplink  association is considered?  The variation of joint rate coverage as a function of downlink and uplink association weights is shown in Fig. \ref{fig:jointpcf} for   three pairs of   ($\pcf, \af$)  with  a rate threshold of $\RATEthresh_u=\RATEthresh_d= 128$ Kbps and $\dnsty{2}= 6 \dnsty{1}$.   As can be seen from the plots, the  uplink and downlink  association weights of  $\frac{\metric{2}}{\metric{1}} =0$ dB, $\frac{\metric{2}^{'}}{\metric{1}^{'}}=-14$ dB ($\metric{1}=\metric{1}^{'} = 1$ in these plots)   also maximize the joint uplink-downlink rate coverage irrespective of chosen $\af$ and $\pcf$\footnote{Other pairs of ($\pcf$, $\af$) also led to similar results.}.   These leads to two key observations: \begin{inparaenum}[(i)] \item the uplink and downlink association weights that maximize  the joint rate coverage are the same as the ones that maximize their individual link  coverage, and as a result \item decoupled association, i.e. different association weights for the   uplink and downlink, is optimal for joint coverage.  \end{inparaenum} 

 \begin{figure}
  \centering
\subfloat[$\pcf=0$, $\af=0.7$]{\includegraphics[width= 0.5\columnwidth]{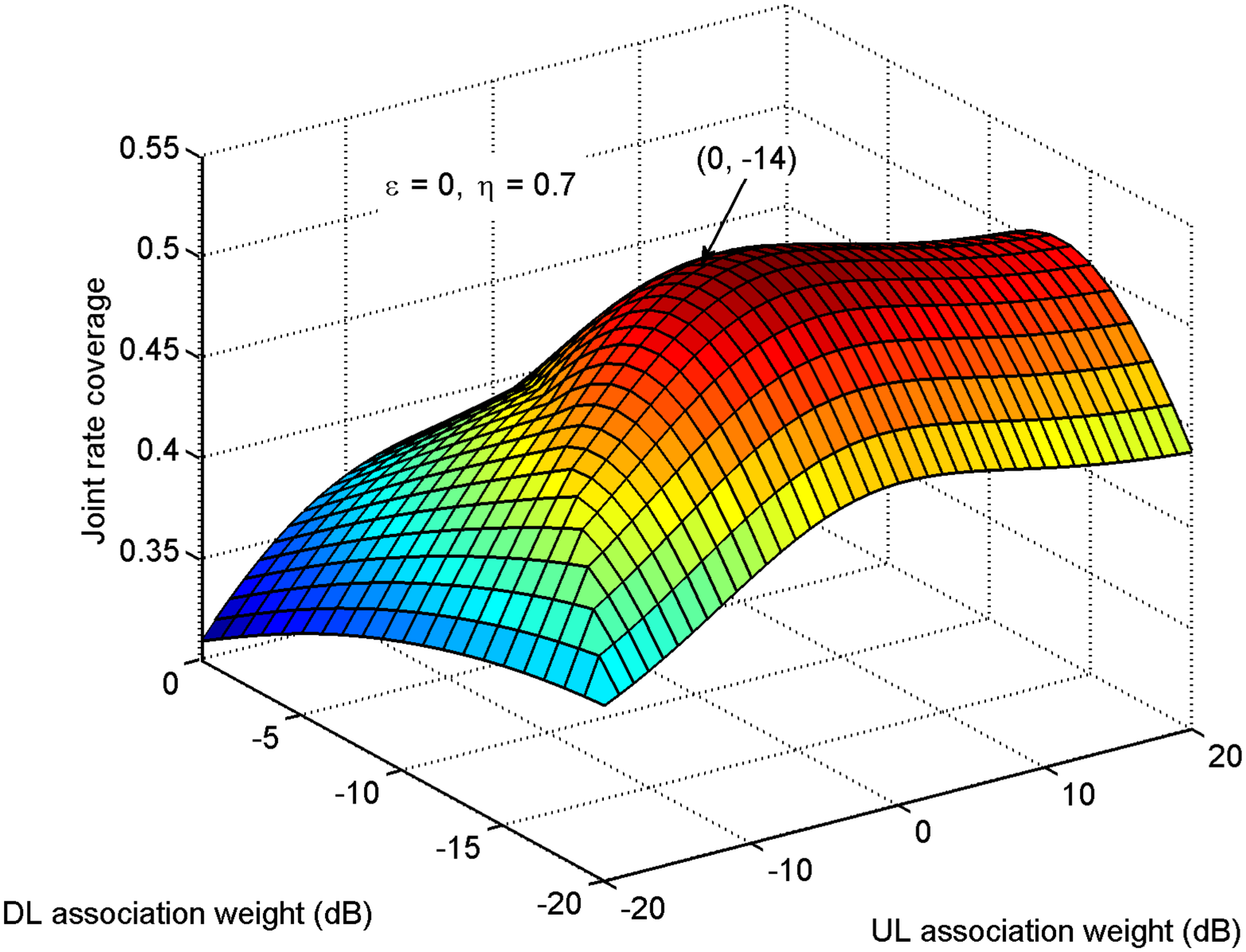}}
\subfloat[$\pcf=1$, $\af=0.3$]{\includegraphics[width= 0.5\columnwidth]{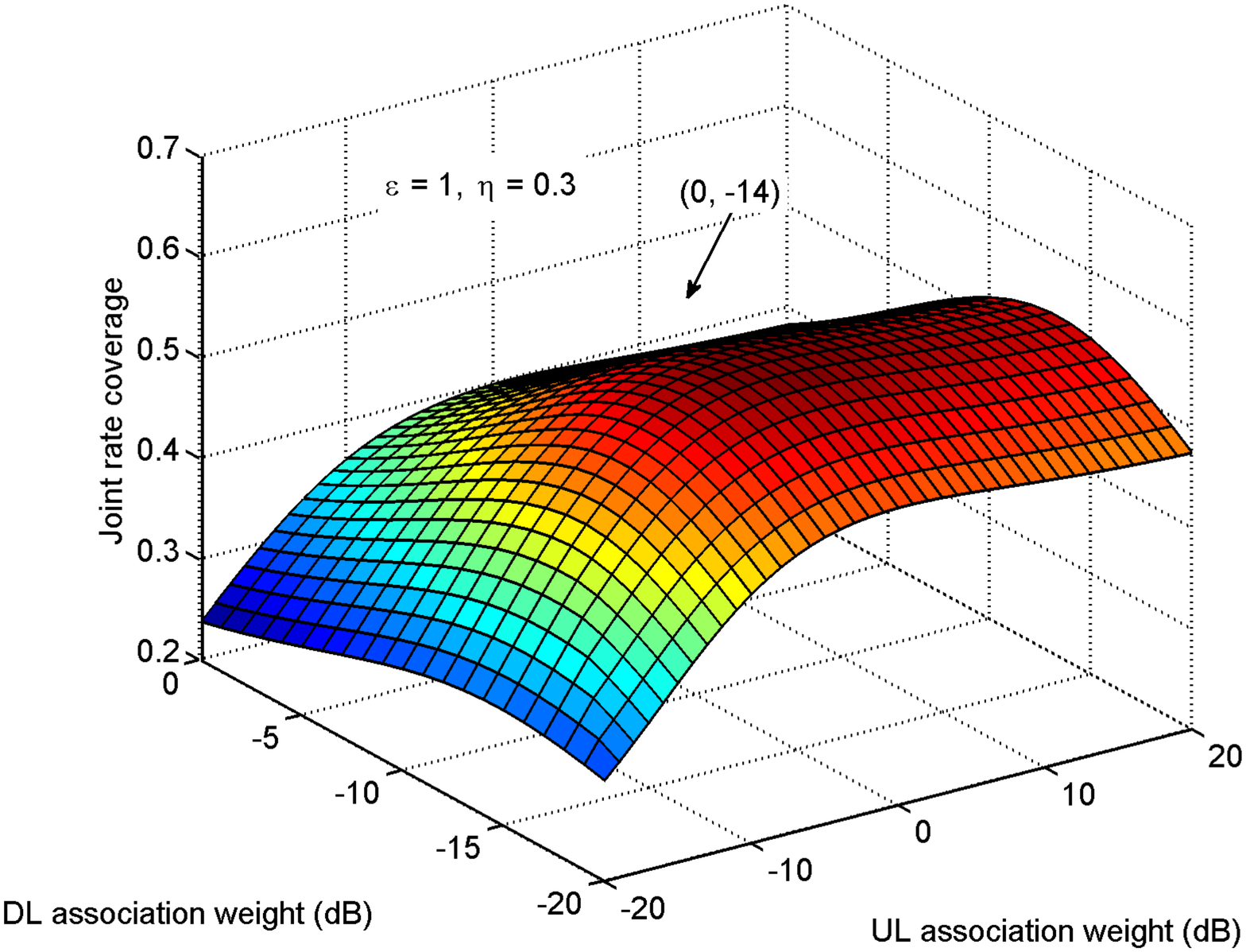}}\\
\subfloat[$\pcf=0.5$, $\af=0.5$]{\includegraphics[width= 0.5\columnwidth]{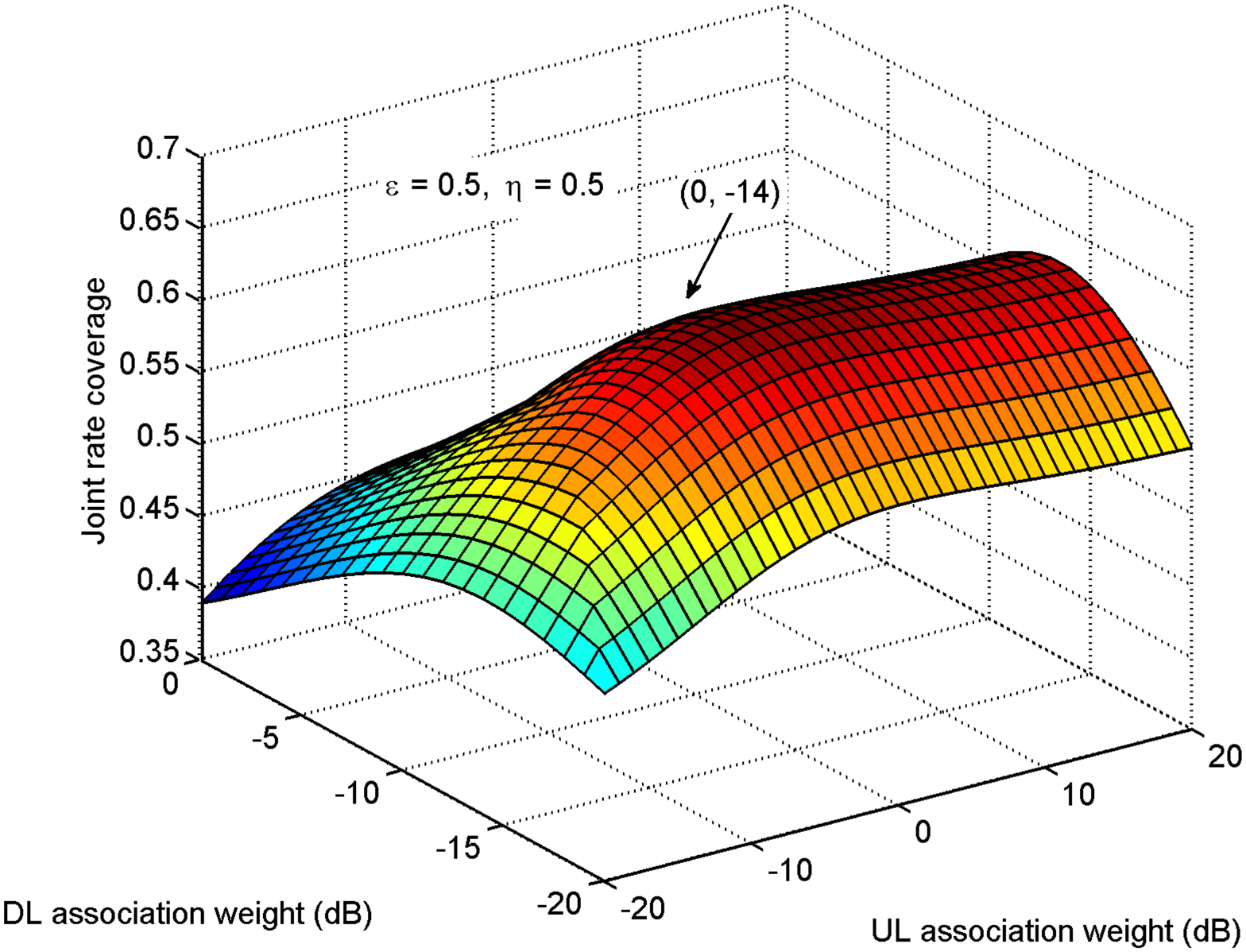}}
\caption{Variation  of joint rate coverage with uplink and downlink association weights (obtained from Theorem \ref{thm:jointrcov}) for different ($\pcf$, $\af$) pairs.}
 \label{fig:jointpcf}
\end{figure}

{\bf Optimal coupled vs. decoupled association.} In Fig. \ref{fig:cdcomp}, the gains of optimal decoupled association over that of coupled are analytically assessed for edge rate and median  rates for varying PCFs with $\dnsty{2}=6\dnsty{1}$ and $\af=0.5$. Note that in these plots, the  rates corresponds the minimum of uplink and downlink, i.e. edge rate $= \rho|\rcov^J(\RATEthresh,\RATEthresh)  =0.95$ and median rate  $= \rho|\rcov^J(\RATEthresh,\RATEthresh)  =0.5$. As observed, across all PCFs, the decoupled association provides significant ($\sim 1.5$x) gain over coupled association. This shows that, in spite requiring  certain architectural changes \cite{SmiPopGavTWC}, decoupled association  is beneficial for applications requiring similar QoS in both uplink and downlink. 

 \begin{figure}
  \centering
\subfloat[Joint edge rate]{\includegraphics[width= 0.5\columnwidth]{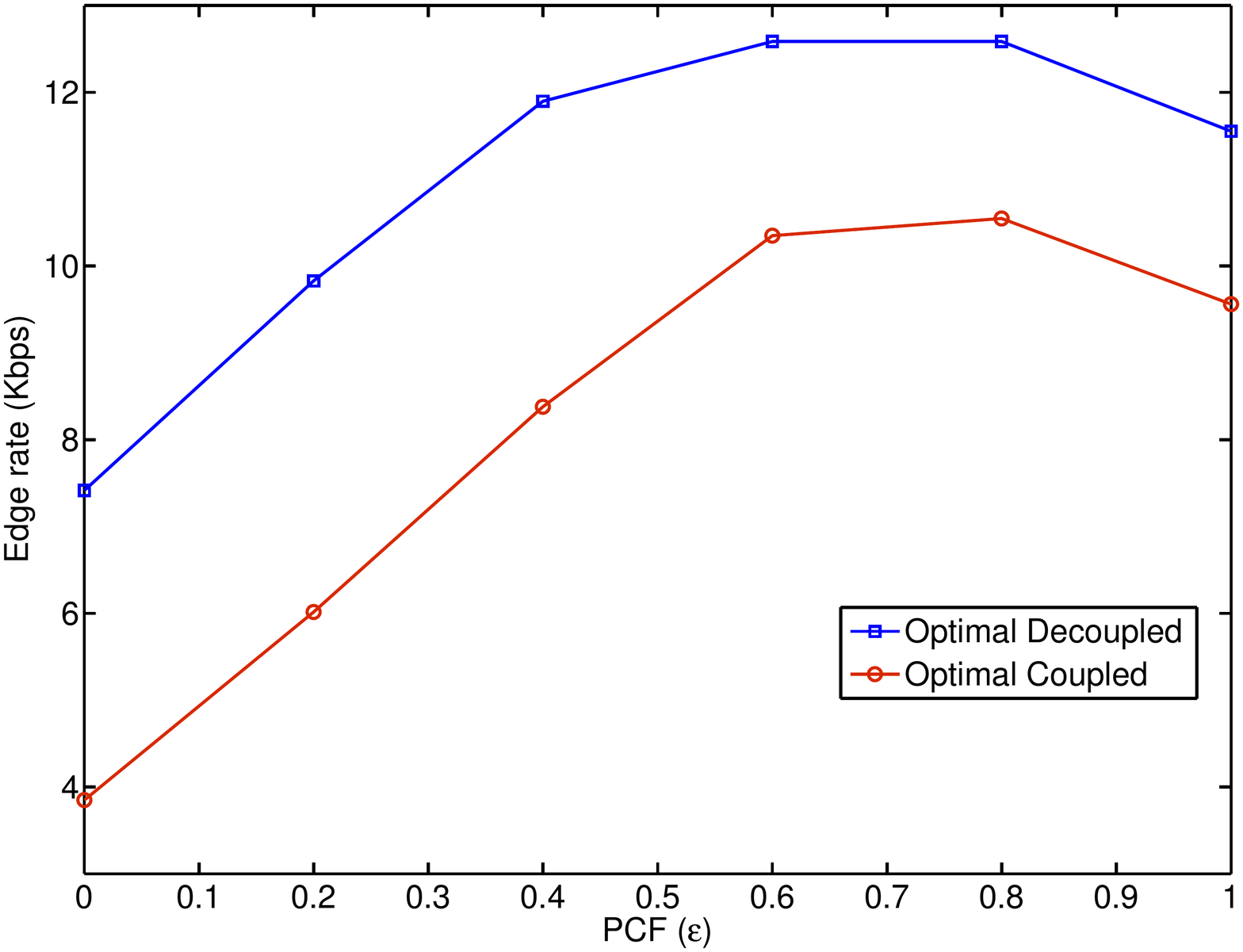}}
\subfloat[Joint median rate]{\includegraphics[width= 0.5\columnwidth]{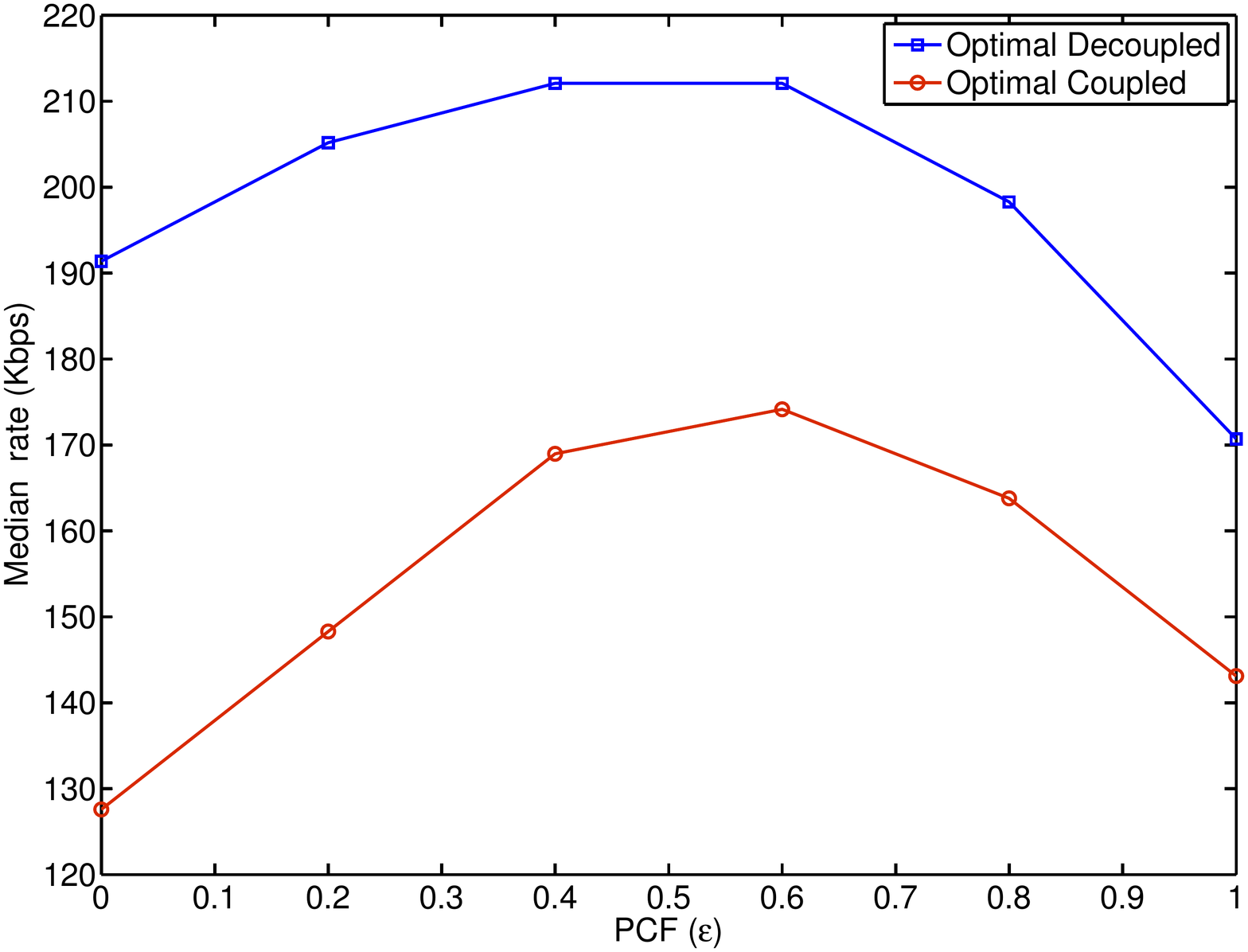}}
\caption{Variation  of uplink-downlink joint edge and median rate with PCF (obtained from Theorem \ref{thm:jointrcov}) for optimal coupled and decoupled   association.}
 \label{fig:cdcomp}
\end{figure}

\section{Conclusion}
This  paper proposes a novel model to analyze  uplink $\SINR$ and rate coverage  in $K$-tier HCNs with load balancing. To the best of the authors' knowledge, this is the first work to derive and validate the uplink rate distribution in HCNs incorporating   offloading and fractional power control. One of the key takeaways from this work is the contrasting behavior exhibited by the uplink and downlink rate distributions with respect to load balancing.  The derivation  of uplink $\SINR$ and  rate distribution as a tractable functional form of   system parameters opens various  areas to gain further design insights. For  example, optimal  association weights were derived in this paper for both uplink and joint uplink-downlink coverage.    We assumed parametric but fixed resource partitioning between uplink and downlink -- and this might also be a more practical assumption -- but analyzing  the impact of more dynamic (possibly load-aware) partitioning on the presented insights  could be considered in the future.  The proposed uplink interference characterization can also be used to analyze systems like massive MIMO, where it plays a crucial role  \cite{MarMIMO}.  Performance analysis for decoupled association incorporating the cost of  possible architectural changes \cite{SmiPopGavTWC} could also be one area of future investigation.

\section*{Acknowledgment}
The authors appreciate helpful feedback from Xingqin Lin. 

\appendices
\section{}\label{sec:prooflapI}
\begin{proof}[Derivation of Lemma \ref{lem:LapI}]
Let $\mgf_{I_{kj}}(s)$ denote the Laplace transform of the interference from tier $j$ UEs, then  $\mgf_{I_k} =\prod_{j=1}^K\mgf_{I_{kj}}$ (from   Assumption \ref{asmptn:indpndnc}). Now, 
\begin{align*}
\mgf_{I_{kj}}(s)& =  \expect{\exp\left(-s \sum_{X\in \PPP_{u,j}^b}\pl_{X}^{ \pcf} H_{X} \pl(X,\cc{})^{-1}\right)} \\
  \overset{(a)}{=} &  \expect{\prod_{X \in \PPP_{u,j}^b}\frac{1}{1+ s \pl_{X}^{\pcf} \pl(X,\cc{ })^{-1}}}\\
  = & \expect{\prod_{X \in \PropPP_{u,j}}\cexpect{\frac{1}{1+ s \pl_X^{\pcf}  X^{-1}}}{\pl_X}} \\
  \overset{(b)}{=} & \exp\left(- \int_{x>0}\left(1-\cexpect{\frac{1}{1+ s\pl_x^{\pcf} x^{-1}}}{\pl_x}\right)\ndnstyr_{u,j}(\mathrm{d}x)\right)\\
 \overset{(c)}{=} & \exp\left(- \int_{x>0}\left(1-\cexpect{\frac{1}{1+ s\pl^{\pcf} x^{-1}}\mid\pl< \nmetric{j}x, \ct{x}=j}{\pl}\right)\ndnstyr_{u,j}(\mathrm{d}x)\right)\\
 = & \exp\left(- \int_{x>0}\cexpect{\frac{1}{1+ (s\pl^{\pcf})^{-1}x} \mid \pl< \nmetric{j}x, \ct{x}=j}{\pl} \ndnstyr_{u,j}(\mathrm{d}x)\right)\\
= & \exp\left(- \cexpect{a_j \pl^{\twople}\nmetric{j}^{-\twople}\int_{1}^{\infty} \frac{\mathrm{d}t}{1+ (s \nmetric{j})^{-1}\pl^{1-\pcf} t^{1/\twople}}}{\pl|\ct{}=j}  \right),
\end{align*}
where  (a) follows from the i.i.d. nature of $\{\chanl_X\}$, (b) follows from the Laplace functional (also known as probability generating functional) of the  assumed PPP $\PropPP_{u,j}$, (c) follows from   Corollary \ref{cor:cdist}, and the last equality follows  with change of variables $t=(x\nmetric{j}/\pl)^{\twople}$ and algebraic manipulation. The final result is then obtained by using the definition of Gauss-Hypergeometric function, yielding
 \begin{equation*}\int_{1}^\infty \frac{\mathrm{d}t}{1+t^{1/\twople}\pl^{1-\pcf}(s\nmetric{j})^{-1}}\\=\frac{\twople}{1-\twople}\frac{s\nmetric{j}}{\pl^{ 1-\pcf }}\,  \hg\left(1,1-\twople,2-\twople,- \frac{s\nmetric{j}}{\pl^{ 1-\pcf}}\right).
 \end{equation*}
\end{proof}

\section{•}\label{sec:proofub}
\begin{proof}[Derivation of  Corollary \ref{cor:ub}]
The  proof of Lemma~\ref{lem:LapI} gives
\begin{align*}
\mgf_{I_{kj}}(s)&= \exp\left(-a_j \nmetric{j}^{-\twople} \int_{l>0} l^{\twople} \int_{1}^{\infty} \frac{\mathrm{d}t}{1+ (s \nmetric{j})^{-1}l^{1-\pcf} t^{1/\twople}}  f_{\pl|\ct{}=j}(l)\mathrm{d}l \right) \\
&=\exp\left(-\frac{a_j \nmetric{j}^{-\twople}}{G_j}    \int_{1}^{\infty} \int_{l>0}\frac{1}{1+ (s \nmetric{j})^{-1}l^{1-\pcf} t^{1/\twople}}  f_{\pl'|\ct{}=j}(l)\mathrm{d}l \mathrm{d}t\right),
\end{align*}
where the inner integral (w.r.t $l$) can be viewed as an expectation of $1/(1+t^{1/\twople} L'^{1-\pcf}(s\nmetric{j})^{-1})$,
where $L_j'$ is a random variable with pdf $f_{\pl'|\ct{}=j}(l) = \twople G_j^2 l^{2\twople-1}\exp(-G_j l^{\twople})$ for $l>0$. Since $1/(1+t^{1/\twople} x (s\nmetric{j})^{-1})$
is a convex function of $x$, we can apply Jensen's inequality (using $\expect{\pl^{'(1-\pcf)}|\ct{}=j}=  \frac{\Gamma(2+(1-\pcf)/\twople)}{G_j^{(1-\pcf)/\twople}}$) and obtain  a lower bound on the inner integral, which leads to an upper
bound on the coverage probability in a form similar to the one in Corollary~\ref{cor:sir}.
\end{proof}

\section{•}\label{sec:prooflb}
\begin{proof}[Derivation of  Corollary \ref{cor:lb}]
Neglecting the conditioning in (c) of the proof of Lemma \ref{lem:LapI}, we have 
\begin{align*}&\mgf_{I_{k}}(s) \geq   \exp\left(-\sum_{j=1}^K\int_{x>0}\cexpect{\frac{1}{1+ (s\pl^{\pcf})^{-1}x}}{\pl|\ct{}=j}\ndnstyr_{u,j}(\mathrm{d}x)\right)\\
& \geq  \exp\left(- \sum_{j=1}^K\cexpect{\int_{x>0}\frac{1}{1+ (s\pl^{\pcf})^{-1}x}\twople\ednsty{j}x^{\twople-1}\mathrm{d}x}{\pl|\ct{}=j}\right)\\
&  \overset{(a)}{=} \exp\left(-  s^{\twople}\frac{\pi\twople}{\sin(\pi\twople)}\sum_{j=1}^K\ednsty{j}\cexpect{\pl^{\twople\pcf}}{\pl|\ct{}=j} \right),
\end{align*} where (a) follows by the change of variables $t = x^{\twople}(s\pl^{\pcf})^{-2/\ple}$ and noting that $\int_{0}^\infty \frac{\mathrm{d}t}{1+t^{\ple/2}}=\frac{2\pi}{\ple\sin(2\pi/\ple)}$.
Now using the coverage expression
\begin{align*}
\pcov(\SINRthresh) &\geq \expect{\exp\left(-\frac{\pi\twople}{\sin(\pi\twople)} \SINRthresh^{\twople}\pl^{\twople(1-\pcf)}\sum_{j=1}^K\ednsty{j} \cexpect{\pl^{\twople\pcf}}{\pl|\ct{}=j}  \right)} \\
&\geq \exp\left(-\frac{\pi\twople}{\sin(\pi\twople)} \SINRthresh^{\twople}\expect{\pl^{\twople(1-\pcf)}}\sum_{j=1}^K\ednsty{j} \cexpect{\pl^{\twople\pcf}}{\pl|\ct{}=j} \right),
\end{align*} where the last inequality follows from Jensen's inequality.
Noting that $\cexpect{\pl^{\twople\pcf}}{\pl|\ct{}=j}= \frac{\Gamma(1+\pcf)}{G_j^\pcf}$ and $\expect{\pl^{\twople(1-\pcf)}} = \sum_{j=1}^K\ednsty{j}\frac{\Gamma(2-\pcf)}{ G_j^{2-\pcf}}$ and $\Gamma(1+\pcf)\Gamma(2-\pcf) = \frac{\pi\pcf(1-\pcf)}{\sin(\pi \pcf)}$ leads to the final result.  
\end{proof}

\bibliographystyle{ieeetr}

\bibliography{IEEEabrv,/Users/sarabjotsingh/Dropbox/research/refoffload}

\end{document}

%% file: notation.tex
\newcommand{\expect}[1]{{\mathbb{E}\left[{#1}\right]}}
\newcommand{\cexpect}[2]{{\mathbb{E}_{#2}\left[{#1}\right]}}
\newcommand{\pr}{{\mathbb{P}}}

\DeclareMathOperator*{\argmin}{arg\,min}

\newcommand{\mgf}{\mathcal{L}}

\newcommand{\loadpmf}{\mathsf{K}}

\newcommand{\hg}{_2\mathrm{F}_1}
\newcommand{\hgc}{\mathrm{C}}

\newcommand{\ple}{\alpha}

\newcommand{\twople}{\delta}

\newcommand{\power}[1]{\mathrm{P}_{#1}}
\newcommand{\res}{\mathrm{W}}

\newcommand{\bias}[1]{\mathrm{B}_{#1}}

\newcommand{\dnsty}[1]{{\lambda_{#1}}}
\newcommand{\ednsty}[1]{{a_{#1}}}

\newcommand{\ndnstyr}{\Lambda}

\newcommand{\userdnsty}{\lambda_u}

\newcommand{\SINRthresh}{\tau}

\newcommand{\RATEthresh}{\rho}

\newcommand{\nRATEthresh}{\hat{\rho}}
 
\newcommand{\uRATEthresh}{v} 
\newcommand{\metric}[1]{\mathrm{T}_{#1}} 
\newcommand{\nmetric}[1]{\hat{\mathrm{T}}_{#1}}

\newcommand{\real}[1]{\mathbb{R}^{#1}}

\newcommand{\uth}[1]{#1^\text{th}}

\newcommand{\ct}[1] {\mathcal{K}_{#1}}
\newcommand{\cc}[1]{\mathcal{B}_{#1}}
\newcommand{\pcf}{\epsilon}

\newcommand{\SINR}{\mathtt{SINR}}
\newcommand{\pl}{L}

\newcommand{\SIR}{\mathtt{SIR}}
\newcommand{\SNR}{\mathtt{SNR}}

\newcommand{\assocr}{\mathcal{C}}
\newcommand{\rate}{\mathtt{Rate}}

\newcommand{\PPP}{\Phi}

\newcommand{\PropPP}{\mathcal{N}}

\newcommand{\PPPu}{\Phi_{u}}

\newcommand{\load}[1]{{N}_{#1}}

\newcommand{\avload}[1]{\bar{N}_{#1}} 
 
\newcommand{\chanl}{H}

\newcommand{\sha}{S}

\newcommand{\indic}{\mathbbm{1}}

\newcommand{\pcov}{\mathcal{P}}

\newcommand{\rcov}{\mathcal{R}}

\newcommand{\avrcov}{\bar{\mathcal{R}}}

\newcommand{\passoc}{\mathcal{A}}

\newcommand{\af}{\eta} 
 
\newtheorem{thm}{{\bf Theorem}}
\newtheorem{cor}{Corollary}
\newtheorem{rem}{Remark}
\newtheorem{lem}{Lemma}
\newtheorem{prop}{Proposition}

\theoremstyle{definition}
\newtheorem{definition}{Definition}

\theoremstyle{thm}
\newtheorem{asmptn}{Assumption}

%% file: ULTWCFirstRevision_v2.bbl
\begin{thebibliography}{10}

\bibitem{ghosh2012heterogeneous}
A.~Ghosh {\em et~al.}, ``Heterogeneous cellular networks: From theory to
  practice,'' {\em {IEEE} Commun. Mag.}, vol.~50, pp.~54--64, June 2012.

\bibitem{qcom_hetnet_wmag}
A.~Damnjanovic {\em et~al.}, ``{A survey on 3GPP heterogeneous networks},''
  {\em {IEEE} Wireless Commun. Mag.}, vol.~18, pp.~10--21, June 2011.

\bibitem{ElSawy13tut}
H.~ElSawy, E.~Hossain, and M.~Haenggi, ``{Stochastic Geometry for Modeling,
  Analysis, and Design of Multi-tier and Cognitive Cellular Wireless Networks:
  A Survey},'' {\em IEEE Communications Surveys \& Tutorials}, vol.~15,
  pp.~996--1019, July 2013.

\bibitem{AndLoadCommag13}
J.~G. Andrews, S.~Singh, Q.~Ye, X.~Lin, and H.~S. Dhillon, ``An overview of
  load balancing in {HetNets}: Old myths and open problems,'' {\em {IEEE}
  Wireless Commun. Mag.}, vol.~21, pp.~18--25, Apr. 2014.

\bibitem{XiaUL}
W.~Xiao {\em et~al.}, ``Uplink power control, interference coordination and
  resource allocation for {3GPP E-UTRA},'' in {\em IEEE Vehicular Technology
  Conference}, Sept. 2006.

\bibitem{mulUL}
R.~M\"{u}llner {\em et~al.}, ``Contrasting open-loop and closed-loop power
  control performance in {UTRAN} {LTE} uplink by {UE} trace analysis,'' in {\em
  IEEE ICC}, pp.~1--6, June 2009.

\bibitem{MuhULPC}
B.~Muhammad and A.~Mohammed, ``Performance evaluation of uplink closed loop
  power control for {LTE} system,'' in {\em IEEE VTC}, pp.~1--5, Sept. 2009.

\bibitem{CasUL}
C.~Castellanos {\em et~al.}, ``Performance of uplink fractional power control
  in {UTRAN LTE},'' in {\em IEEE VTC}, pp.~2517--2521, May 2008.

\bibitem{SimULPC}
A.~Simonsson and A.~Furuskar, ``Uplink power control in {LTE} - overview and
  performance,'' in {\em IEEE VTC}, pp.~1--5, Sept. 2008.

\bibitem{NovUL}
T.~Novlan, H.~Dhillon, and J.~Andrews, ``Analytical modeling of uplink cellular
  networks,'' {\em {IEEE} Trans. Wireless Commun.}, vol.~12, pp.~2669--2679,
  June 2013.

\bibitem{Cou2011}
M.~Coupechoux and J.-M. Kelif, ``How to set the fractional power control
  compensation factor in {LTE} ?,'' in {\em Proc. IEEE Sarnoff Symposium},
  pp.~1--5, May 2011.

\bibitem{andganbac11}
J.~G. Andrews, F.~Baccelli, and R.~K. Ganti, ``A tractable approach to coverage
  and rate in cellular networks,'' {\em {IEEE} Trans. Commun.}, vol.~59,
  pp.~3122--3134, Nov. 2011.

\bibitem{BlaKarKee12}
B.~Blaszczyszyn, M.~K. Karray, and H.-P. Keeler, ``Using {Poisson} processes to
  model lattice cellular networks,'' in {\em Proc. IEEE Intl. Conf. on Comp.
  Comm. (INFOCOM)}, pp.~773--781, Apr. 2013.

\bibitem{ADG_Tcom}
A.~Guo and M.~Haenggi, ``Asymptotic deployment gain: A simple approach to
  characterize the $\mathtt{SINR}$ distribution in general cellular networks,''
  {\em {IEEE} Trans. Commun.}, 2014.
\newblock Submitted, available at: http://arxiv.org/abs/1404.6556.

\bibitem{SinDhiAnd13}
S.~Singh, H.~S. Dhillon, and J.~G. Andrews, ``Offloading in heterogeneous
  networks: Modeling, analysis, and design insights,'' {\em {IEEE} Trans.
  Wireless Commun.}, vol.~12, pp.~2484--2497, May 2013.

\bibitem{Soellerhaus}
B.~B{\l}aszczyszyn and D.~Yogeshwaran, ``Clustering comparison of point
  processes with applications to random geometric models,'' 2012.
\newblock accepted for {\em Stochastic Geometry, Spatial Statistics and Random
  Fields: Analysis, Modeling and Simulation of Complex Structures}, (V.
  Schmidt, ed.) Lecture Notes in Mathematics Springer. Available at
  http://arxiv.org/abs/1212.5285.

\bibitem{ElSawyH14}
H.~ElSawy and E.~Hossain, ``On stochastic geometry modeling of cellular uplink
  transmission with truncated channel inversion power control,'' {\em {IEEE}
  Trans. Wireless Commun.}, vol.~13, pp.~4454--4469, Aug. 2014.

\bibitem{LeeUL14}
H.~Lee, Y.~Sang, and K.~Kim, ``On the uplink {SIR} distributions in
  heterogeneous cellular networks,'' {\em IEEE Commun. Let.}, vol.~to appear,
  Oct. 2014.

\bibitem{ElshaerBDI14}
H.~Elshaer, F.~Boccardi, M.~Dohler, and R.~Irmer, ``Downlink and uplink
  decoupling: a disruptive architectural design for {5G} networks,'' in {\em
  IEEE Global Commun. Conf. (GLOBECOM)}, Dec. 2014.

\bibitem{SmiPopGavTWC}
K.~Smiljkovikj, H.~Elshaer, P.~Popovski, F.~Boccardi, M.~Dohler,
  L.~Gavrilovska, and R.~Irmer, ``Capacity analysis of decoupled downlink and
  uplink access in {5G} heterogeneous systems,'' {\em {IEEE} Trans. Wireless
  Commun.}, 2014.
\newblock Submitted. Available at: http://arxiv.org/abs/1410.7270.

\bibitem{SmiPopGav}
K.~Smiljkovikj, P.~Popovski, and L.~Gavrilovska, ``Analysis of the decoupled
  access for downlink and uplink in wireless heterogeneous networks,'' {\em
  IEEE Wireless Comm. Lett.}, 2014.
\newblock Submitted. Avaialble at: arxiv.org/abs/1407.0536.

\bibitem{SinAnd14}
S.~Singh and J.~G. Andrews, ``Joint resource partitioning and offloading in
  heterogeneous cellular networks,'' {\em {IEEE} Trans. Wireless Commun.},
  vol.~13, pp.~888--901, Feb. 2014.

\bibitem{josanxiaand12}
H.-S. Jo, Y.~J. Sang, P.~Xia, and J.~G. Andrews, ``Heterogeneous cellular
  networks with flexible cell association: A comprehensive downlink
  $\mathtt{SINR}$ analysis,'' {\em {IEEE} Trans. Wireless Commun.}, vol.~11,
  pp.~3484--3495, Oct. 2012.

\bibitem{dhiganbacand12}
H.~S. Dhillon, R.~K. Ganti, F.~Baccelli, and J.~G. Andrews, ``Modeling and
  analysis of {$K$}-tier downlink heterogeneous cellular networks,'' {\em
  {IEEE} J. Sel. Areas Commun.}, vol.~30, pp.~550--560, Apr. 2012.

\bibitem{SinBacAnd13}
S.~Singh, F.~Baccelli, and J.~G. Andrews, ``On association cells in random
  heterogeneous networks,'' {\em IEEE Wireless Commun. Lett.}, vol.~3,
  pp.~70--73, Feb. 2014.

\bibitem{MadHCN12}
P.~Madhusudhanan, J.~Restrepo, Y.~Liu, and T.~Brown, ``Downlink coverage
  analysis in a heterogeneous cellular network,'' in {\em IEEE Global Commun.
  Conf. (GLOBECOM)}, pp.~4170--4175, Dec. 2012.

\bibitem{JinUL}
N.~Jindal, S.~Weber, and J.~Andrews, ``Fractional power control for
  decentralized wireless networks,'' {\em {IEEE} Trans. Wireless Commun.},
  vol.~7, no.~12, pp.~5482--5492, 2008.

\bibitem{BacLi11}
F.~Baccelli, J.~Li, T.~Richardson, S.~Shakkottai, S.~Subramanian, and X.~Wu,
  ``On optimizing {CSMA} for wide area ad-hoc networks,'' in {\em Intl. Symp.
  on Modeling and Optimization in Mobile, Ad Hoc and Wireless Networks
  (WiOpt)}, pp.~354--359, May 2011.

\bibitem{ye2012user}
Q.~Ye, B.~Rong, Y.~Chen, M.~Al-Shalash, C.~Caramanis, and J.~G. Andrews, ``User
  association for load balancing in heterogeneous cellular networks,'' {\em
  {IEEE} Trans. Wireless Commun.}, vol.~12, pp.~2706--2716, June 2013.

\bibitem{MarMIMO}
T.~Marzetta, ``Noncooperative cellular wireless with unlimited numbers of base
  station antennas,'' {\em {IEEE} Trans. Wireless Commun.}, vol.~9,
  pp.~3590--3600, Nov. 2010.

\end{thebibliography}
